\documentclass[12pt,draftclsnofoot,onecolumn]{IEEEtran}

\ifCLASSINFOpdf
  
\else
  
\fi

\usepackage{color}
\usepackage{setspace}
\usepackage{cite}
\usepackage{graphicx}
\usepackage{subfig}
\usepackage{amsfonts}
\usepackage{amsmath}
\usepackage{amsthm}
\usepackage{mathdots}
\usepackage{bm}
\usepackage{caption}
\newtheorem{theorem}{Theorem}

\newtheorem{lemma}{Lemma}

\begin{document}

\begin{spacing}{1.36}

\title{Achievable Rate {\color{black}Analysis} and Phase Shift Optimization on Intelligent Reflecting Surface with Hardware Impairments}

\author{Zhe Xing,~\IEEEmembership{Student Member,~IEEE,}
        Rui Wang,~\IEEEmembership{Senior Member,~IEEE,}
        Jun Wu,~\IEEEmembership{Senior Member,~IEEE,}
        and Erwu Liu,~\IEEEmembership{Senior Member,~IEEE}
\thanks{The authors are with the College of Electronics and Information Engineering, Tongji University, Shanghai 201804, China (e-mail: zxing@tongji.edu.cn; ruiwang@tongji.edu.cn; wujun@tongji.edu.cn; erwuliu@tongji.edu.cn).}
}

\markboth{}%
{\MakeLowercase{\textit{et al.}}: Achievable Rate Analysis and Phase Shift Optimization on Intelligent Reflecting Surface with Hardware Impairments}

\maketitle

\begin{abstract}
Intelligent reflecting surface (IRS) is envisioned as a promising hardware solution to hardware cost and energy consumption in the fifth-generation (5G) mobile communication network. It exhibits great advantages in enhancing data transmission, but may suffer from performance degradation caused by inherent hardware impairment (HWI). For analysing the achievable rate (ACR) and optimizing the phase shifts in the IRS-aided wireless communication system with HWI, we consider that the HWI appears at both the IRS and the signal transceivers. {\color{black}On this foundation}, first, we derive the closed-form expression of the average ACR {\color{black} and the IRS utility}. Then, we formulate optimization problems to optimize the IRS phase shifts by maximizing the signal-to-noise ratio (SNR) at the receiver side, and obtain the solution by transforming non-convex problems into semidefinite programming (SDP) problems. Subsequently, we compare the IRS with the conventional decode-and-forward (DF) relay {\color{black} in terms of the ACR and the utility}. Finally, we carry out simulations to verify the theoretical {\color{black}analysis, and evaluate the impact of the channel estimation errors and residual phase noises on the optimization performance}. Our results reveal that {\color{black} the HWI reduces the ACR and the IRS utility, and begets more serious performance degradation with more reflecting elements.} Although the HWI has an impact on the IRS, it still leaves opportunities for the IRS to surpass the conventional DF relay, when the number of reflecting elements is large enough {\color{black} or the transmitting power is sufficiently high}.
\end{abstract}


\IEEEpeerreviewmaketitle

\newpage
\section{Introduction}
\IEEEPARstart{T}{he} rapid development of the worldwide mobile communication technologies has been witnessed in recent years. After the 4th generation (4G) mobile communications became universal around the world, the initial 5th generation (5G) standard was completed in 2018 and the 5G commercial networks were already employed in part in the first quarter of 2020. For supporting huge mobile data traffic and high-speed communications required by a growing number of {\color{black}the} mobile devices accessed to the wireless networks, a variety of innovative techniques including millimeter wave (mmWave), ultra-dense network (UDN) and massive multiple-input multiple-output (MIMO) are implemented in 5G wireless transmission systems \cite{M.Agiwal2016(CST)}. These techniques exhibit great advantages in helping the communication systems improve spectral efficiency (SE) \cite{W.Yan2019(WCL)}, but {\color{black}face} challenging problems such as: 1) the mmWave is susceptible to blockage and suffers from severe power attenuation during the long-distance propagation in the atmosphere \cite{X.Lin2015(WCL)}, so that the wireless communication system will bear poor reliability when the received signals are substantially weak; 2) the UDN is composed of numerous intensively distributed base stations (BS) \cite{W.Sun2018(TWC)} while the massive MIMO requests the signal transceivers to be equipped with large-scale antenna arrays \cite{S.Hu2018(TSP)}, which {\color{black}lead} to high hardware cost (HWC).

One mature technological solution to these problems is utilizing relays to establish a multi-hop transmission mode. Conventional wireless cooperative communication systems mostly employ relays \cite{S.Cheng2018(TVT), P.K.Sharma2016(CL), X.Xia2015(TSP), A.K.Mishra2018(TVT)} to process on the signals received halfway and retransmit the signals to the destination terminals actively through an uncontrollable propagation environment. Relays are validated to be effective on improving system reliability \cite{P.K.Sharma2016(CL)}, but are still active retransmitting facilities that require high energy consumption (EC) and HWC. Recently, another state of the art approach, which is named Intelligent Reflecting Surface (IRS) \cite{O.Ozdogan2019(WCL), Q.Wu2019(ICASSP)}, Large Intelligent Surface (LIS) \cite{W.Yan2019(WCL)} or Large Intelligent Metasurface (LIM) \cite{Z.He2020(WCL)}, has attracted {\color{black}considerable attention} from wireless communication researchers. An IRS is a planar array composed of a large number of low-cost passive reconfigurable reflecting elements, each of which induces an adjustable phase shift on the coming signal wave and reflects the signal to the destination terminal \cite{Q.Wu2019(TWC), C.Huang2019(JSAC), C.Huang2020(WCM), Liuyiming(arXiv)}. It is {\color{black}distinct} from the ordinary physical reflecting surfaces which simply reflect the signal waves without any parameter adjustment, and also different from the traditional relays which actively retransmit the received signals. As a passive reflecting apparatus, the IRS is envisioned as a promising hardware solution to EC and HWC in the future communication networks.

There have already been studies that focused on the achievable rate (ACR) maximization, energy efficiency improvement, modulation scheme, secure communication realization, phase shift optimization, channel estimation and capacity {\color{black}analysis for} the IRS-aided wireless communication systems \cite{E.Basar2020(TCOM), M.Cui2019(WCL), H.Shen2019(CL), W.Yan2019(WCL), Q.Nadeem2020(arXiv), E.Bjornson2020(WCL), C.Huang2018(ICASSP), C.Huang2019(TWC)}. For instance, C. Huang, \textit{et al}. \cite{C.Huang2018(ICASSP), C.Huang2019(TWC)} employed the IRS to maximize the ACR \cite{C.Huang2018(ICASSP)} and the energy efficiency \cite{C.Huang2019(TWC)} of the wireless communication systems. E. Basar \cite{E.Basar2020(TCOM)} proposed an IRS-based index modulation scheme which enabled high data rate and low bit-error-rate (BER). M. Cui, \textit{et al}. \cite{M.Cui2019(WCL)} and H. Shen, \textit{et al}. \cite{H.Shen2019(CL)} developed IRS-aided secure wireless communication systems where the IRS was employed to maximize the rate gap (secrecy rate) between the desired transmission path from the source to the legitimate user and the undesired one from the source to the eavesdropper. W. Yan, \textit{et al}. \cite{W.Yan2019(WCL)} developed a passive beamforming and information transferring method and optimized the phase shifts with different state values to improve the average signal-to-noise ratio (SNR). Q. Nadeem, \textit{et al}. \cite{Q.Nadeem2020(arXiv)} outlined an IRS-aided multiple-user MIMO communication system and estimated the cascaded channel matrix within each time interval. E. Björnson, \textit{et al}. \cite{E.Bjornson2020(WCL)} analysed and compared the channel capacities of the IRS-supported, the decode-and-forward (DF) relay assisted and the single-input-single-output (SISO) communication systems, and derived the least required number of {\color{black}the} IRS reflecting elements which allowed {\color{black}the} IRS to outperform the DF relay and SISO. 

It is noted that the aforementioned works are carried out under the assumption of perfect hardware. However, in most practical situations, the inherent hardware impairment (HWI), e.g. phase noise, quantization error, amplifier non-linearity, \textit{et al}., which {\color{black}will} generally limit the system performance, cannot be neglected due to the non-ideality of the communication devices in the real world \cite{X.Zhang2015(TCOM), K.Xu2015(IEEE PACRIM), X.Xia2015(TSP)}. Although the effect of the HWI on the system performance can be mitigated by compensation algorithms \cite{T.Schenk2008(book)}, there will still exist residual HWI due to the imprecisely estimated time-variant hardware characteristic and the random noise. As a result, it is of great significance to probe into the system performance in the presence of HWI. Some researchers \cite{E.Bjornson2014(TIT), Q.Zhang2018(TWC), E.Bjornson2013(CL)} analysed the channel capacity of the massive MIMO communication systems with HWI, which they modelled as additive Gaussian distributed distortion noise. But to the best of our knowledge, there were only a few studies that analysed the IRS-aided communication systems with the HWI at the IRS \cite{S.Hu2018(GlobalCom), M.A.Badiu2020(WCL),D.Li2020(CL)}. Among these studies, the researchers performed some important initial works by modelling {\color{black}the HWI at the IRS} as an additive variable {\color{black}with respect to} the distance between {\color{black}the} reflecting point and the reflecting surface center \cite{S.Hu2018(GlobalCom)}, or as the uniformly distributed phase noise generated by the reflecting units \cite{M.A.Badiu2020(WCL),D.Li2020(CL)}. However, {\color{black} these works still left several research gaps to be filled. First,} the HWI of the transmitting devices and receiving terminals was not simultaneously taken into consideration, which would jointly influence the performance of the IRS-aided communication systems as well. {\color{black} Second, the phase shift optimization was not implemented when there existed HWI, which was indispensable for one to acquire the optimal IRS configuration with hardware imperfections. Third, the performance comparisons between the IRS and the conventional approaches, e.g. DF relay, which also contributed to the wireless data transmission enhancement, needed to be further explored in the presence of HWI.} Up to now, we have not found the {\color{black} related works that inquired into the above three aspects}. Therefore, in this article, we will provide the ACR {\color{black}analysis} and phase shift optimization on the IRS-aided communication system in consideration of {\color{black}the HWI at both the IRS and transceivers}, and present performance comparisons with the existing multiple-antenna DF relay assisted {\color{black}communication} system with the HWI at the DF relay and transceivers. Our contributions are summarized as follows.

\begin{itemize}
\item[•] By referring to \cite{M.A.Badiu2020(WCL)}, we model the {\color{black}HWI at the IRS} as a phase error matrix, in which the random phase errors generated by the IRS reflecting units are uniformly distributed. By referring to \cite{E.Bjornson2014(TIT),E.Bjornson2015(TWC)}, {\color{black}we model the transceiver HWI as the additive distortion noise as well as the phase drift and the thermal noise}. When the IRS phase shifts are adjusted to compensate for the phase shifts in the {\color{black}source-IRS} channel and the {\color{black}IRS-destination} channel, we mathematically derive the closed-form expression of the average ACR {\color{black}and the IRS utility} with respect to the number of the reflecting elements, denoted by $N$. From the theoretical and the numerical results, we confirm that {\color{black}the HWI decreases the average ACR and the IRS utility, and imposes more severe impact on the ACR performance as $N$ becomes larger.}

\item[•] In order to optimize the IRS phase shifts and obtain the maximum average ACR with {\color{black}HWI}, we formulate the optimization problems and transform the non-convex problems into convex semidefinite programming (SDP) problems, then obtain the solution numerically by exploiting CVX toolbox with SDPT3 solver in the MATLAB simulation. {\color{black}Besides, we evaluate the impact of the channel estimation errors and the residual phase noises on the optimization performance, after which we conclude that both of the two unavoidable factors result in performance degradation to some extent.}

\item[•] {\color{black}When the HWI appears at the IRS, the DF relay and the transceivers}, we compare the performance of the IRS with that of the {\color{black}multiple-antenna} DF relay in terms of the ACR {\color{black}and the utility}, and derive the condition where the IRS can always surpass the DF relay for all $N>0$. The results illustrate that if $N$ is large enough {\color{black}or the transmitting power is sufficiently high}, the IRS with $N$ passive reflecting elements is able to outperform the DF relay with the same number of antennas in the presence of HWI.

\end{itemize}

The rest of this article is organized as follows. In Section II, we introduce the IRS-aided {\color{black}wireless} communication system with HWI by showing the system model. In Section III, we analyse the ACR {\color{black}and the IRS utility} in the considered wireless communication system. In Section IV, we {\color{black}narrate the problem formulation and transformation} when optimizing the IRS phase shifts {\color{black}in the presence of HWI}. In Section V, we compare the performance of the IRS with that of the multiple-antenna DF relay in terms of the ACR {\color{black}and the utility}. In Section VI, we provide numerical results to verify the theoretical {\color{black}analysis and present discussions on the channel estimation errors and the residual phase noises}. In Section VII, we draw the overall conclusions.

\textit{Notations}: Italics denote the variables or constants, while boldfaces denote the vectors or matrices. $\mathbf{A}^*$, $\mathbf{A}^T$, $\mathbf{A}^H$ and $ \mathbf{A}^{-1}$ symbolize the conjugate, transpose, conjugate-transpose and inverse of matrix $\mathbf{A}$, respectively.  $tr(\mathbf{A})$ and $rank(\mathbf{A})$ stand for the trace and the rank of $\mathbf{A}$. $diag(\mathbf{a})$ represents an $n\times n$ sized diagonal matrix whose diagonal elements are $(a_1,a_2,\ldots,a_n)$ in vector $\mathbf{a}$. $||.||_2$ represents $\ell_2$ norm. {\color{black}$\odot$ symbolizes the Hadamard product.} $\mathbf{A}\in\mathbb{C}^{m\times n}$ or $\mathbf{A}\in\mathbb{R}^{m\times n}$ means that $\mathbf{A}$ is an $m\times n$ sized complex or real-number matrix. $\mathbf{A}\sim\mathcal{CN}(\mathbf{0},\mathbf{V})$ or $\mathbf{A}\sim\mathcal{N}(\mathbf{0},\mathbf{V})$ illustrates that $\mathbf{A}$ obeys complex normal or normal distribution with mean of zero and covariance matrix of $\mathbf{V}$. $\mathbf{A}\succeq\mathbf{0}$ means that $\mathbf{A}$ is positive semidefinite. $\mathbb{E}_\mathbf{x}[\mathbf{A}]$ denotes the expectation of $\mathbf{A}$ on the random variable $\mathbf{x}$ if $\mathbf{A}$ is a stochastic matrix in relation to $\mathbf{x}$. $\mathbf{I}_n$ and $\mathbf{\Gamma}_n$ symbolize $n\times n$ sized identity matrix and $n\times n$ sized matrix with all elements of 1, respectively. $\mathbf{1}$ stands for the unit row vector with all elements of 1. $\Delta=b^2-4ac$ represents the discriminant of the quadratic function $f\left(x\right)=ax^2+bx+c$. {\color{black}$g(x)=\mathcal{O}(f(x))$ signifies that $|g(x)/f(x)|$ is bounded when $x\rightarrow\infty$. $\lim_{x\to\infty}f(x)$ is represented by $f(x)|_{x\rightarrow\infty}$ throughout the whole paper.}

\section{Communication System Model}

In this article, the considered wireless communication system (Figure \ref{Fig-1}) includes a signal-emitting source (e.g.  the base station, BS), an IRS with $N$ passive reflecting elements, an IRS controller and a signal-receiving destination (e.g. the user equipment, UE). The signal-emitting source, assumed to be equipped with single antenna, transmits the modulated signals with an average signal power of $\sqrt{P}$. The IRS induces reconfigurable phase shifts, which are adjusted by the IRS controller based on the channel state information (CSI), on the impinging signals, and reflects the coming signal waves to the destination. The signal-receiving destination, also equipped with single antenna, receives the directly arrived signals from the source and passively reflected signals from the IRS.

\begin{figure}[!t]
\includegraphics[width=4.0in]{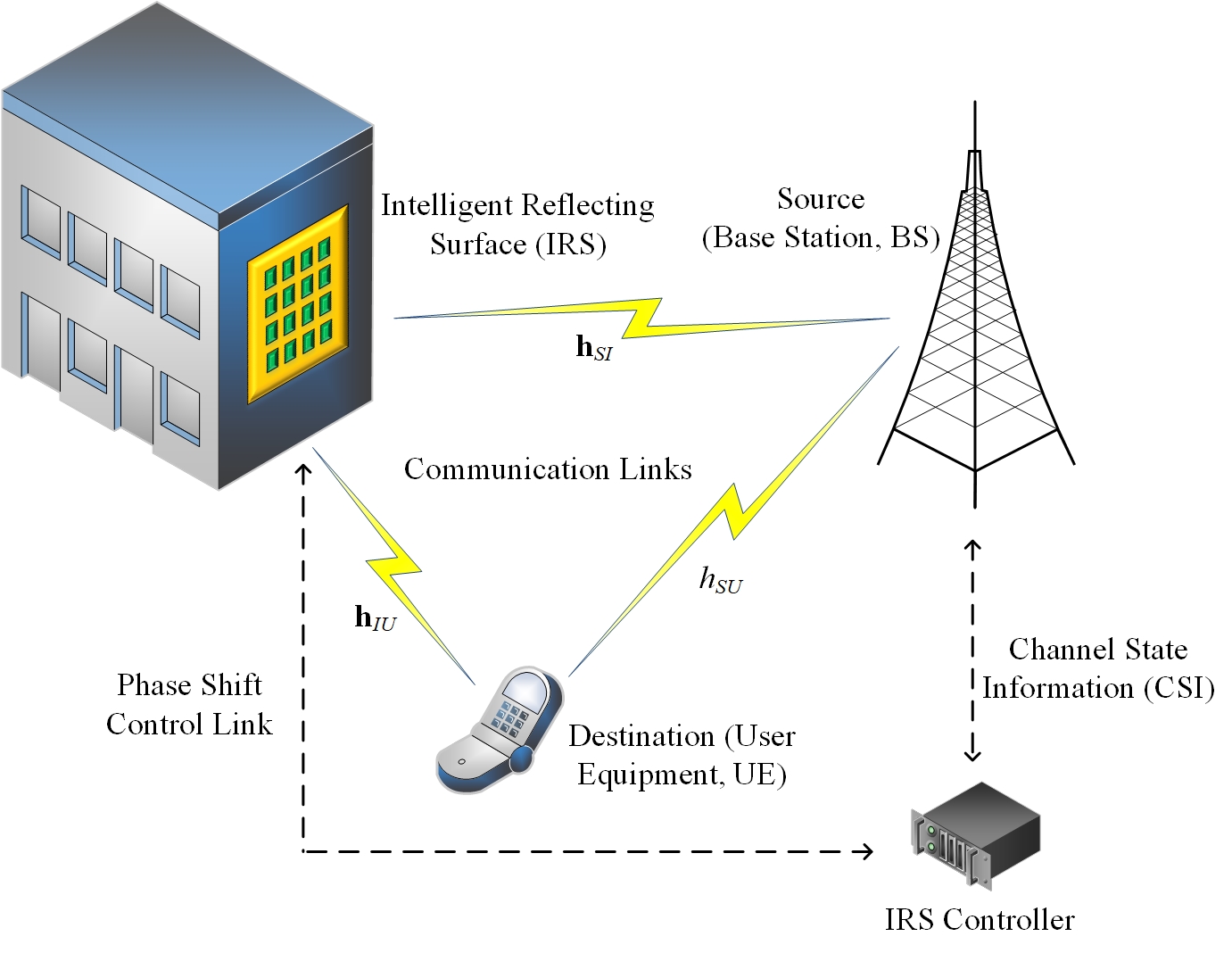}
\centering
\caption{The considered IRS-aided wireless communication system, including a single-antenna signal-emitting source, a single-antenna signal-receiving destination, an IRS with $N$ passive reflecting elements, and an IRS controller.}
\label{Fig-1}
\end{figure}

Generally, due to the non-ideality of the hardware, the received signal is disturbed by the HWI which universally exists in the real-world communication devices. In this considered system, {\color{black} the HWI appears at both the IRS and the signal transceivers}. First, {\color{black} the HWI at the IRS} is modelled as a random diagonal phase error matrix, which involves $N$ random phase errors induced by {\color{black} the} intrinsic hardware {\color{black} imperfection} of the passive reflectors, or by {\color{black} the} imprecision of the channel estimation \cite{M.A.Badiu2020(WCL)}. It is expressed as 
{\color{black}\begin{equation}
\mathbf{\Theta}_E=diag\left(e^{j\theta_{E1}},e^{j\theta_{E2}},\ldots,e^{j\theta_{EN}}\right)
\end{equation}
where} {\color{black}$j^2=-1$}; $\theta_{Ei}$, for $i=1,2,\ldots,N$, are random phase errors uniformly distributed on $\left[-\pi/2,\pi/2\right]$.
{\color{black} Then, the HWIs at the signal transceivers primarily include the additive distortion noise, the multiplicative phase drift and the amplified thermal noise \cite{E.Bjornson2013(CL), E.Bjornson2014(TIT), E.Bjornson2015(TWC), Q.Zhang2018(TWC)}, which create a mismatch between the intended signal and the practically generated signal, or create a distortion on the received signal during the reception processing. The distortion noises, generated by the transmitter and the receiver due to the insufficiency of the accurate modelling, the time-variant characteristics, \textit{et al}., are modelled as $\eta_t(t)\sim\mathcal{CN}(0,\Upsilon_t)$ and $\eta_r(t)\sim\mathcal{CN}(0,V_r)$, respectively, where $\Upsilon_t$ and $V_r$ will be given in (\ref{Upsilon_t}) and (\ref{V_r}).
The multiplicative phase drift caused by the local oscillator at the receiver is modelled as $\phi(t)=e^{j\psi(t)}$, with its expectation given by $\mathbb{E}[\phi(t)]=e^{-\frac{1}{2}\delta t}$ \cite{Q.Zhang2018(TWC)}, where $\delta$ denotes the oscillator quality, and $\psi(t)$ follows a Wiener process:
\begin{equation}
\psi(t)\sim\mathcal{N}(\psi(t-1),\delta)
\end{equation}
The amplified thermal noise, aroused by the mixers at the receiver and by the interference leakage from other frequency bands or wireless networks \cite{E.Bjornson2015(TWC)}, is modelled as $w'(t)\sim\mathcal{CN}(0,\sigma_{w'}^2)$, with $\sigma_{w'}^2$ being the thermal noise variance. }

Therefore, referring to Eq. (2) in \cite{E.Bjornson2014(TIT)} and {\color{black} Eq. (3) in \cite{E.Bjornson2015(TWC)}}, the received signal disturbed by {\color{black} HWI} is modelled as
{\color{black}\begin{equation}\label{eq2-12}
y(t)=\phi(t)\left(\mathbf{h}_{IU}^T\mathbf{\Phi}\mathbf{\Theta}_E\mathbf{h}_{SI}+h_{SU}\right)\left[\sqrt Ps(t)+\eta_t(t)\right]+\eta_r(t)+w(t)
\end{equation}
where} $s(t)$ stands for the unit-power signal symbol at time $t$ with $\mathbb{E}\left[s(t)s^*(t)\right]=1$; $w(t)\sim\mathcal{CN}\left(0,\sigma_w^2\right)$ denotes the receiver noise, whose {\color{black} variance $\sigma_w^2$, according to \cite{E.Bjornson2015(TWC)}, satisfies $\sigma_w^2=F\sigma_{w'}^2$, with $F>1$ being the noise amplification factor}; $\mathbf{\Phi}=\alpha\times diag\left(e^{j\theta_1},e^{j\theta_2},\ldots,e^{j\theta_N}\right)$ represents the phase shifting matrix of the IRS, where $\alpha\in(0,1]$ is the fixed amplitude reflection coefficient and $\theta_i$, for $i=1,2,\ldots,N$, are the adjustable phase-shift variables of the IRS; $h_{SU}=\sqrt{\mu_{SU}}e^{j\varphi_{SU}}$ represents the channel coefficient from the source to the destination, where $\sqrt{\mu_{SU}}$ and $\varphi_{SU}$ are the power attenuation coefficient and the phase shift of $h_{SU}$; $\mathbf{h}_{IU}\in\mathbb{C}^{N\times 1}$ and $\mathbf{h}_{SI}\in\mathbb{C}^{N\times1}$ are the channel coefficients from the IRS to the destination and from the source to the IRS, respectively, which are expressed as{\color{black}\cite{E.Bjornson2020(WCL)}}
\begin{equation}
\mathbf{h}_{IU}=\sqrt{\mu_{IU}}\left(e^{j\varphi_{IU,1}},e^{j\varphi_{IU,2}},\ldots,e^{j\varphi_{IU,N}}\right)^T
\end{equation}
\begin{equation}
\mathbf{h}_{SI}=\sqrt{\mu_{SI}}\left(e^{j\varphi_{SI,1}},e^{j\varphi_{SI,2}},\ldots,e^{j\varphi_{SI,N}}\right)^T
\end{equation}
where $\sqrt{\mu_{IU}}$ and $\sqrt{\mu_{SI}}$ are the power attenuation coefficients of $\mathbf{h}_{IU}$ and $\mathbf{h}_{SI}$; $\varphi_{IU,i}$ and $\varphi_{SI,i}$, for $i=1,2,\ldots,N$, are the phase shifts of $\mathbf{h}_{IU}$ and $\mathbf{h}_{SI}$.
{\color{black} As the distortion noises are proportional to the signal power, we have
\begin{equation}\label{Upsilon_t}
\Upsilon_t=\kappa_tP \mathbb{E}\left[s(t)s^*(t)\right]
\end{equation}
\begin{equation}\label{V_r}
V_r=\kappa_rP\left|\phi(t)\left(\mathbf{h}_{IU}^T\mathbf{\Phi}\mathbf{\Theta}_E\mathbf{h}_{SI}+h_{SU}\right)\right|^2 \mathbb{E}\left[s(t)s^*(t)\right]
\end{equation}
where} $\kappa_t$ and $\kappa_r$ represent the proportionality coefficients which describe the severities of the distortion noises at the transmitter and the receiver, respectively.

For this communication system, we {\color{black} will} {\color{black} derive the approximate closed-form ACR expression and the IRS utility in relation to $N$ in the presence of HWI, and analyse the ACR and utility degradations caused by HWI in the next section}.

\section{ACR Analysis with HWI}
Based on the signal model in {\color{black} (\ref{eq2-12})}, we {\color{black} will} analyse the ACR {\color{black} and the IRS utility} of the considered IRS-aided communication system in the presence of {\color{black} HWI}.
Here we assume that the phase information in the cascaded source-IRS-destination channel model \cite{Q.Nadeem2020(arXiv)} is already estimated before $\mathbf{\Phi}$ is adjusted, so that $\left(\varphi_{IU,i}+\varphi_{SI,i}\right)$, for $i=1,2,\ldots,N$, are known for the IRS phase shift controller. This can be realized via some existing channel estimation techniques in e.g.  \cite{B.Zheng2020(WCL),Q.Nadeem2020(arXiv)}, which estimated the cascaded channel, {\color{black}and \cite{L.Wei2020(Arxiv),L.Wei2020(IEEE SAM)}, which designed robust and effective channel estimation frameworks based on the PARAllel FACtor (PARAFAC)}. In {\color{black} (\ref{eq2-12})}, $\mathbf{h}_{IU}^T\mathbf{\Phi}\mathbf{\Theta}_E\mathbf{h}_{SI}$ is maximized if each phase shift of the IRS is adjusted into $\theta_i=-\left(\varphi_{IU,i}+\varphi_{SI,i}\right)$, for $i=1,2,\ldots,N$, to compensate for the phase shifts in $\mathbf{h}_{IU}$ and $\mathbf{h}_{SI}$ \cite{E.Bjornson2020(WCL)}. As a result, when $\theta_i=-\left(\varphi_{IU,i}+\varphi_{SI,i}\right)$, the received signal affected by {\color{black} HWI} is expressed as
{\color{black} \begin{equation}\label{eq2-7}
y(t)=\phi(t)\left(\alpha\mathbf{g}_{IU}^T\mathbf{\Theta}_E\mathbf{g}_{SI}+h_{SU}\right)\left[\sqrt Ps(t)+\eta_t(t)\right]+\eta_r(t)+w(t)
\end{equation}
where} $\mathbf{g}_{IU}=\sqrt{\mu_{IU}}\mathbf{1}^T$ and $\mathbf{g}_{SI}=\sqrt{\mu_{SI}}\mathbf{1}^T$. {\color{black}Accordingly}, the ACR with {\color{black} HWI} is expressed as
{\color{black} \begin{equation}\label{eq2-16}
\begin{split}
R_{HWI}\left(N\right)&=\log_2\left\{1+\frac{P\left|\phi(t)\left(\alpha\mathbf{g}_{IU}^T\mathbf{\Theta}_E\mathbf{g}_{SI}+h_{SU}\right)\right|^2}{P(\kappa_t+\kappa_r)\left|\phi(t)\left(\alpha\mathbf{g}_{IU}^T\mathbf{\Theta}_E\mathbf{g}_{SI}+h_{SU}\right)\right|^2+\sigma_w^2}\right\}
\end{split}
\end{equation}}

Based on (\ref{eq2-16}), we obtain the following theorem.

{\color{black} \begin{theorem}
When $\theta_i=-\left(\varphi_{IU,i}+\varphi_{SI,i}\right)$ and $\theta_{Ei}$ is uniformly distributed on $\left[-\pi/2,\pi/2\right]$, the approximate average ACR with {\color{black} HWI} is expressed as
\begin{equation}\label{eq2-17}
\overline{R_{HWI}}\left(N\right)=\log_2\left\{1+\frac{\beta N^2+\lambda N+\mu_{SU}}
{\left(\kappa_t+\kappa_r\right)\left(\beta N^2+\lambda N+\mu_{SU}\right)+\frac{\sigma_w^2}{P}}\right\}
\end{equation}
where
\begin{equation}\label{eq-revision-1}
\beta=\frac{4\alpha^2\mu_{IU}\mu_{SI}}{\pi^2}
\end{equation}
\begin{equation}\label{eq-revision-2}
\lambda=\left(1-\frac{4}{\pi^2}\right)\alpha^2\mu_{IU}\mu_{SI}+\frac{4\alpha}{\pi}\mu_{IU}^{\frac{1}{2}}\mu_{SI}^{\frac{1}{2}}\mu_{SU}^{\frac{1}{2}}\cos{(\varphi_{SU})}
\end{equation}

The IRS utility with HWI, defined by $\gamma_{HWI}(N)=\frac{\partial\overline{R_{HWI}}\left(N\right)}{\partial N}$ according to the Definition 1 in \cite{S.Hu2018(GlobalCom)}, is expressed as
\begin{equation}\label{Utility Expression}
\begin{split}
\gamma_{HWI}(N)=\frac{\sigma_w^2}{P}(2\beta N+\lambda)\left\{\left[(\kappa_t+\kappa_r)(\beta N^2+\lambda N+\mu_{SU})+\frac{\sigma_w^2}{P}\right]\times\right.\\
\left. \left[(\kappa_t+\kappa_r+1)(\beta N^2+\lambda N+\mu_{SU})+\frac{\sigma_w^2}{P}\right]\ln 2\right\}^{-1}
\end{split}
\end{equation}
\end{theorem}
\begin{proof}
The proof is given in Appendix A.
\end{proof}}

{\color{black} Subsequently, for theoretically evaluating the impact that the HWI has on the ACR and the IRS utility, we further calculate the rate gap, defined by $\delta_R(N)=R(N)-\overline{R_{HWI}}\left(N\right)$, and the utility gap, defined by $\delta_{\gamma}(N)=\gamma(N)-\gamma_{HWI}(N)$ in the following \textbf{Lemma 1}, where $R(N)$ and $\gamma(N)$, denoting the ACR and the IRS utility without HWI, will be given in the proof.}

\begin{lemma}
When $\theta_i=-\left(\varphi_{IU,i}+\varphi_{SI,i}\right)$ and $\theta_{Ei}$ is uniformly distributed on $\left[-\pi/2,\pi/2\right]$, the rate gap {\color{black}$\delta_R\left(N\right)$} between the average ACR{\color{black}s} with and without HWI is expressed as
{\color{black}\begin{equation}\label{eq2-18}
\delta_R\left(N\right)\!=\!\log_2\left\{\frac{P\left(\kappa_t+\kappa_r\right)\chi+\sigma_w^2+P^2\chi\varpi\left(\frac{\kappa_t}{\sigma_w^2}+\frac{\kappa_r}{\sigma_w^2}\right)+P\varpi}{P(\kappa_t+\kappa_r+1)\chi+\sigma_w^2}\right\}
\end{equation}
where 
\begin{equation}\label{eq-revision-3}
\varpi=\frac{\pi^2}{4}\beta N^2+\rho N+\mu_{SU}
\end{equation}
\begin{equation}\label{eq-revision-4}
\chi=\beta N^2+\lambda N+\mu_{SU}
\end{equation}
with $\rho$ given by $\rho=2\alpha\mu_{IU}^{\frac{1}{2}}\mu_{SI}^{\frac{1}{2}}\mu_{SU}^{\frac{1}{2}}\cos{(\varphi_{SU})}$.

The utility gap $\delta_{\gamma}(N)$ between the IRS utilities with and without HWI is expressed as
\begin{equation}\label{Utility Degradation}
\begin{split}
\delta_\gamma(N)=&
\left[P^3\chi^2(\kappa_t+\kappa_r+1)\left(\frac{\kappa_t}{\sigma_w^2}+\frac{\kappa_r}{\sigma_w^2}\right)\frac{\partial\varpi}{\partial N}+
P^2(\kappa_t+\kappa_r+1)\left(\frac{\partial\varpi}{\partial N}\chi-\frac{\partial\chi}{\partial N}\varpi\right)+\right.\\
&\left. P^2\left(\kappa_t+\kappa_r\right)\left(\frac{\partial\varpi}{\partial N}\chi+\frac{\partial\chi}{\partial N}\varpi\right)+
P\sigma_w^2\left(\frac{\partial\varpi}{\partial N}-\frac{\partial\chi}{\partial N}\right)\right]\times
\left\{
\left[P(\kappa_t+\kappa_r)\chi+\sigma_w^2+\right.\right.\\
&\left.\left.P^2\varpi\chi\left(\frac{\kappa_t}{\sigma_w^2}+\frac{\kappa_r}{\sigma_w^2}\right)+P\varpi\right]
\left[P(\kappa_t+\kappa_r+1)\chi+\sigma_w^2\right]\ln 2
\right\}^{-1}
\end{split}
\end{equation}
where $\frac{\partial\varpi}{\partial N}=\frac{\pi^2}{2}\beta N+\rho$ and $\frac{\partial\chi}{\partial N}=2\beta N+\lambda$ are the partial derivatives of $\varpi$ and $\chi$, respectively.}
\end{lemma}
{\color{black} \begin{proof}
According to \cite{E.Bjornson2020(WCL)}, $R\left(N\right)$ is expressed as 
\begin{equation}\label{eq2-11}
R\left(N\right)=\log_2\left\{1+\frac{P\left[\alpha^2N^2\mu_{IU}\mu_{SI}+2\alpha N\sqrt{\mu_{IU}\mu_{SI}\mu_{SU}}\cos{\left(\varphi_{SU}\right)}+\mu_{SU}\right]}{\sigma_w^2}\right\}
\end{equation}

Then, the corresponding $\gamma(N)$, defined by $\gamma(N)=\frac{\partial R(N)}{\partial N}$, is given by
\begin{equation}\label{Utility Expression 2}
\begin{split}
\gamma(N)&=\frac{\frac{P}{\sigma_w^2}\left[2\alpha^2\mu_{IU}\mu_{SI}N+2\alpha\sqrt{\mu_{IU}\mu_{SI}\mu_{SU}}\cos{\left(\varphi_{SU}\right)}\right]}
{\left\{1+\frac{P}{\sigma_w^2}\left[\alpha^2 N^2\mu_{IU}\mu_{SI}+2\alpha N\sqrt{\mu_{IU}\mu_{SI}\mu_{SU}}\cos{\left(\varphi_{SU}\right)}+\mu_{SU}\right]\right\}\ln 2}
\end{split}
\end{equation}

Thereupon, by calculating $\delta_R(N)=R(N)-\overline{R_{HWI}}\left(N\right)$ and $\delta_{\gamma}(N)=\gamma(N)-\gamma_{HWI}(N)=\frac{\partial R(N)}{\partial N}-\frac{\partial \overline{R_{HWI}}(N)}{\partial N}=\frac{\partial \delta_R(N)}{\partial N}$, we derive the above (\ref{eq2-18}) and (\ref{Utility Degradation}).
\end{proof}}

{\color{black}\textbf{Theorem 1} demonstrates that: 1) although the $\overline{R_{HWI}}\left(N\right)$ increases with $N$, the proportionality coefficient $\beta$ on $N^2$ in $\overline{R_{HWI}}\left(N\right)$ is smaller than $\alpha^2\mu_{IU}\mu_{SI}$ in $R(N)$, hinting that $\overline{R_{HWI}}\left(N\right)$ rises more slowly than $R(N)$. 
2) When $N\rightarrow\infty$, the $\overline{R_{HWI}}\left(N\right)$ is limited by
\begin{equation}\label{R_HWI Upper Bound}
\left.\overline{R_{HWI}}\left(N\right)\right|_{N\rightarrow\infty}=\log_2\left(1+\frac{1}{\kappa_t+\kappa_r}\right)
\end{equation}
which signifies that even if $N$ becomes significantly large or tends to be infinite, the potential growth of $\overline{R_{HWI}}\left(N\right)$ will be primarily restricted by $\kappa_t$ and $\kappa_r$ of the HWI at the signal transceivers. On the contrary, $R(N)$ continuously ascends without bound as $N$ grows. 
3) The $\gamma_{HWI}(N)$ is inversely proportional to $\mathcal{O}(N^3)$, which indicates that the IRS utility with HWI descends as $N$ grows, and is close to zero when $N\rightarrow\infty$. This implies that if $N$ is extremely large or nearly infinite, adding more passive reflecting elements will contribute to little ACR improvement when there exists HWI.

\textbf{Lemma 1} illustrates that: 1) the rate gap $\delta_R(N)>0$ for $N>0$, which indicates that the ACR is degraded by HWI. 2) The $\delta_R(N)$ increases with $N$, because the numerator inside $\log_2(.)$ contains $\chi\varpi$ which is proportional to $\mathcal{O}(N^4)$, while the denominator inside $\log_2(.)$ merely involves $\chi$ which is proportional to $\mathcal{O}(N^2)$. This implies that as $N$ grows, the IRS-aided wireless communication system will suffer from more serious ACR degradation. 3) The utility gap $\delta_{\gamma}(N)>0$, because by expanding $\left(\frac{\partial\varpi}{\partial N}\chi-\frac{\partial\chi}{\partial N}\varpi\right)$ and $\left(\frac{\partial\varpi}{\partial N}-\frac{\partial\chi}{\partial N}\right)$ in (\ref{Utility Degradation}), we have
\begin{equation}
\begin{split}
&\left(\frac{\partial\varpi}{\partial N}\chi-\frac{\partial\chi}{\partial N}\varpi\right)=
\frac{4\alpha^2\mu_{IU}\mu_{SI}N^2}{\pi^2}\left[\left(\frac{\pi^2}{4}-1\right)\alpha^2\mu_{IU}\mu_{SI}+(\pi-2)\alpha\mu_{IU}^{\frac{1}{2}}\mu_{SI}^{\frac{1}{2}}\mu_{SU}^{\frac{1}{2}}\cos{(\varphi_{SU})} \right]+\\
&\left[\left(2-\frac{8}{\pi^2}\right)N+\frac{4}{\pi^2}-1\right]\alpha^2\mu_{IU}\mu_{SI}\mu_{SU}+
\left(2-\frac{4}{\pi}\right)\alpha\mu_{IU}^{\frac{1}{2}}\mu_{SI}^{\frac{1}{2}}\mu_{SU}^{\frac{3}{2}}\cos{(\varphi_{SU})}>0
\end{split}
\end{equation}
\begin{equation}
\begin{split}
\left(\frac{\partial\varpi}{\partial N}-\frac{\partial\chi}{\partial N}\right)=(2N-1)\left(1-\frac{4}{\pi^2}\right)\alpha^2\mu_{IU}\mu_{SI}+\left(2-\frac{4}{\pi}\right)\alpha\mu_{IU}^{\frac{1}{2}}\mu_{SI}^{\frac{1}{2}}\mu_{SU}^{\frac{1}{2}}\cos{(\varphi_{SU})}>0
\end{split}
\end{equation}

This reveals that the IRS utility will be also degraded by HWI to some extent.}

It is notable that the results in \textbf{Theorem 1} and \textbf{Lemma 1} are derived on the basis of $\theta_i=-\left(\varphi_{IU,i}+\varphi_{SI,i}\right)$, which is configured to compensate for the phase shifts in $\mathbf{h}_{IU}$ and $\mathbf{h}_{SI}$ \cite{E.Bjornson2020(WCL)}. However, $\theta_i=-\left(\varphi_{IU,i}+\varphi_{SI,i}\right)$ might not be optimal in this considered wireless propagation environment, as it does not take the phase shift in $h_{SU}$ into account. Thus, in Section IV, we will optimize the IRS phase shifts and reconfigure the phase shifting matrix to obtain the maximum ACR in the presence of HWI.

\section{Phase Shift Optimization}
Instead of configuring $\theta_i=-\left(\varphi_{IU,i}+\varphi_{SI,i}\right)$ to evaluate the ACR, we {\color{black} will} formulate the optimization problem to optimize the IRS phase shifts with {\color{black} HWI} in this section. 

{\color{black}\subsection{Problem Formulation and Transformation}}

Here, we retrospect {\color{black} (\ref{eq2-12})}, from which we obtain the ACR with {\color{black} HWI}:
{\color{black} \begin{equation}\label{original ACR with HWI}
\begin{split}
R_{\mathbf{\Phi},HWI}\left(N\right)&=\log_2\left\{1+\frac{P\left|\phi(t)\left(\mathbf{h}_{IU}^T\mathbf{\Phi}\mathbf{\Theta}_E\mathbf{h}_{SI}+h_{SU}\right)\right|^2}{P(\kappa_t+\kappa_r)\left|\phi(t)\left(\mathbf{h}_{IU}^T\mathbf{\Phi}\mathbf{\Theta}_E\mathbf{h}_{SI}+h_{SU}\right)\right|^2+\sigma_w^2}\right\}
\end{split}
\end{equation}

Therefore, aiming at maximizing the received SNR, we can formulate the phase shift optimization problem as
\begin{subequations}\label{original problem}
\begin{align}
(\mathrm{P1}):\ \mathop{\max}\limits_{\mathbf{\Phi}}&\frac{P\left|\phi(t)\left(\mathbf{h}_{IU}^T\mathbf{\Phi}\mathbf{\Theta}_E\mathbf{h}_{SI}+h_{SU}\right)\right|^2}{P(\kappa_t+\kappa_r)\left|\phi(t)\left(\mathbf{h}_{IU}^T\mathbf{\Phi}\mathbf{\Theta}_E\mathbf{h}_{SI}+h_{SU}\right)\right|^2+\sigma_w^2}
\\s.t.\ &\left|\left[\mathbf{\Phi}\right]_{(n,n)}\right|=\alpha,\ n=1,2\ldots N
\end{align}
\end{subequations}

However, the objective function (OBF) in (\ref{original problem}a) is non-concave with respect to $\mathbf{\Phi}$, and the constraint in (\ref{original problem}b) is non-convex. Thus, inspired by \cite{M.Cui2019(WCL)}, we will convert (P1) into another solvable form.}


Let $\mathbf{D}_{IU}$ denote a diagonal matrix expressed as $\mathbf{D}_{IU}=diag\left(\mathbf{h}_{IU}\right)$, and $\bm{\theta}$ denote a column vector expressed as $\bm{\theta}=\alpha\left(e^{j\theta_1},e^{j\theta_2},\ldots,e^{j\theta_N}\right)^T$. Then, we have $\bm{\theta}^T\mathbf{D}_{IU}=\mathbf{h}_{IU}^T\mathbf{\Phi}$. By replacing $\mathbf{h}_{IU}^T\mathbf{\Phi}$ with $\bm{\theta}^T\mathbf{D}_{IU}$, we expand (\ref{original ACR with HWI}) into
{\color{black} \begin{equation}\label{eq2-22}
R_{\bm{\theta},HWI}\left(N\right)=\log_2{\left\{1+\frac{P\left(Z+||h_{SU}||_2^2\right)}{P(\kappa_t+\kappa_r)\left(Z+||h_{SU}||_2^2\right)+\sigma_w^2}\right\}}
\end{equation}
where}
$Z=\mathbf{h}_{SI}^H\mathbf{\Theta}_E^H\mathbf{D}_{IU}^H\bm{\theta}^*\bm{\theta}^T\mathbf{D}_{IU}\mathbf{\Theta}_E\mathbf{h}_{SI}+\mathbf{h}_{SI}^H\mathbf{\Theta}_E^H\mathbf{D}_{IU}^H\bm{\theta}^*h_{SU}+h_{SU}^*\bm{\theta}^T\mathbf{D}_{IU}\mathbf{\Theta}_E\mathbf{h}_{SI}$.

Let $\mathbf{a}$ be defined by $\mathbf{a}=\left(\bm{\theta}^T,1\right)^H$. We can rewrite $Z$ as $Z=\mathbf{a}^H\mathbf{\Xi}\mathbf{a}$, where
\begin{equation}\label{eq2-24}
\mathbf{\Xi}=\left(\begin{matrix}\mathbf{D}_{IU}\mathbf{\Theta}_E\mathbf{h}_{SI}\mathbf{h}_{SI}^H\mathbf{\Theta}_E^H\mathbf{D}_{IU}^H&h_{SU}^*\mathbf{D}_{IU}\mathbf{\Theta}_E\mathbf{h}_{SI}\\\mathbf{h}_{SI}^H\mathbf{\Theta}_E^H\mathbf{D}_{IU}^Hh_{SU}&0\\\end{matrix}\right)
\end{equation}

Therefore, {\color{black} $R_{\bm{\theta},HWI}\left(N\right)$} can be simplified into
{\color{black}\begin{equation}\label{eq2-25}
\begin{split}
R_{\bm{\theta},HWI}\left(N\right)=&\log_2{\left\{1+\frac{P\left(\mathbf{a}^H\mathbf{\Xi}\mathbf{a}+||h_{SU}||_2^2\right)}{P(\kappa_t+\kappa_r)\left(\mathbf{a}^H\mathbf{\Xi}\mathbf{a}+||h_{SU}||_2^2\right)+\sigma_w^2}\right\}} \\
=&\log_2{\left\{1+\frac{P\left[tr(\mathbf{\Xi}\mathbf{X})+||h_{SU}||_2^2\right]}{P(\kappa_t+\kappa_r)\left[tr(\mathbf{\Xi}\mathbf{X})+||h_{SU}||_2^2\right]+\sigma_w^2}\right\}}
\end{split}
\end{equation}
where }
\begin{equation}\label{eq2-26}
\mathbf{X}=\mathbf{a}\mathbf{a}^H=\left(\begin{matrix}\bm{\theta}^*\bm{\theta}^T&\bm{\theta}^*\\\bm{\theta}^T&1\\\end{matrix}\right)\in\mathbb{C}^{\left(N+1\right)\times\left(N+1\right)}
\end{equation}

Then, the optimization problem is formulated as
{\color{black}\begin{subequations}\label{converted problem 1}
\begin{align}
(\mathrm{P2}):\ \mathop{\max}\limits_{\bm{\theta}}&\frac{P\left[tr(\mathbf{\Xi}\mathbf{X})+||h_{SU}||_2^2\right]}{P(\kappa_t+\kappa_r)\left[tr(\mathbf{\Xi}\mathbf{X})+||h_{SU}||_2^2\right]+\sigma_w^2}
\\s.t.\ &\left|\left[\bm{\theta}\right]_{n}\right|=\alpha,\ n=1,2\ldots N
\end{align}
\end{subequations}
which is still non-convex due to the complicated non-concave OBF in (\ref{converted problem 1}a) and the non-convex module constraint in (\ref{converted problem 1}b).

Here, from $\bm{\theta}^*\bm{\theta}^T$ in (\ref{eq2-26}), it can be realized that the diagonal entries in $\mathbf{X}$ embody the modules of the elements in $\bm{\theta}$. Thus, we define a simple matrix $\mathbf{E}_n$, whose $(i,j)$-th element is given by
\begin{equation}\label{eq2-28}
\left[\mathbf{E}_n\right]_{(i,j)}=\left\{\begin{matrix}1,\ \ \ \ \ i=j=n\\0,\ \ \ \ otherwise\\\end{matrix}\right.
\end{equation}

As a result, the optimization problem is converted into
\begin{subequations}\label{converted problem 2}
\begin{align}
(\mathrm{P3}):\ \mathop{\max}\limits_{\mathbf{X}\succeq\mathbf{0}}&\frac{P\left[tr(\mathbf{\Xi}\mathbf{X})+||h_{SU}||_2^2\right]}{P(\kappa_t+\kappa_r)\left[tr(\mathbf{\Xi}\mathbf{X})+||h_{SU}||_2^2\right]+\sigma_w^2}
\\s.t.\ &tr\left(\mathbf{E}_n\mathbf{X}\right)=\alpha^2,\ n=1,2\ldots N
\\&tr\left(\mathbf{E}_{N+1}\mathbf{X}\right)=1
\\&rank(\mathbf{X})=1
\end{align}
\end{subequations}
where the constraint in (\ref{converted problem 2}b) is transformed from the module constraint of $\left|\left[\bm{\theta}\right]_n\right|^2=\left[\bm{\theta}^*\bm{\theta}^T\right]_{(n,n)}=\mathbf{a}^H\mathbf{E}_n\mathbf{a}=tr\left(\mathbf{E}_n\mathbf{X}\right)=\alpha^2$, for $n=1,2\ldots N$; the constraint in (\ref{converted problem 2}c) is transformed from $\left[\mathbf{X}\right]_{(N+1,N+1)}=1$; the constraint in (\ref{converted problem 2}d) is responsible for strictly guaranteeing that 1) the resolved $\mathbf{X}$ can be decomposed into $\mathbf{X}=\mathbf{a}\mathbf{a}^H$, and 2) the solution of the phase shift in $\bm{\theta}$ in the resolved $\mathbf{X}$ is equivalent to the solution of the phase shift in $\mathbf{\Phi}$ in (P1).

Nevertheless, (P3) is still non-convex and is difficult to solve. Therefore, the problem transformation should be further performed. Thanks to the Charnes-Cooper transformation \cite{L.Liu2014(TSP), Charnes1962}}, let $\mathbf{Y}$ and $\mu$ be defined by $\mathbf{Y}=\mu\mathbf{X}$ and
$\mu=\frac{1}{tr(\mathbf{\Xi}\mathbf{X})+||h_{SU}||_2^2+\frac{\sigma_w^2}{P\left(\kappa_t+\kappa_r\right)}}$.
Then, the OBF in (\ref{converted problem 2}a) is expressed as
$\frac{1}{\left(\kappa_t+\kappa_r\right)}\times\left[tr(\mathbf{\Xi}\mathbf{Y})+\mu||h_{SU}||_2^2\right]$.
Therefore, (P3) is transformed into

\begin{subequations}\label{converted problem 3}
\begin{align}
(\mathrm{P4}):\ \mathop{\max}\limits_{\mathbf{Y}\succeq\mathbf{0},\mu\geq 0}&\frac{1}{\left(\kappa_t+\kappa_r\right)}\times\left[tr(\mathbf{\Xi}\mathbf{Y})+\mu||h_{SU}||_2^2\right]
\\s.t.\ &tr\left(\mathbf{E}_n\mathbf{Y}\right)=\mu\alpha^2,\ n=1,2\ldots N
\\&tr\left(\mathbf{E}_{N+1}\mathbf{Y}\right)=\mu
\\&tr\left(\mathbf{\Xi Y}\right)+\mu\left[||h_{SU}||_2^2+\frac{\sigma_w^2}{P\left(\kappa_t+\kappa_r\right)}\right]=1
\\&rank(\mathbf{Y})=1
\end{align}
\end{subequations}

{\color{black} Although (P4) is still non-convex due to the non-convex (\ref{converted problem 3}e), it can be relaxed if the constraint of $rank(\mathbf{Y})=1$ is omitted. Hence, the relaxed problem is formulated as}
\begin{subequations}\label{eq2-43}
\begin{align}
(\mathrm{P5}):\ \mathop{\max}\limits_{\mathbf{Y}\succeq\mathbf{0},\mu\geq 0}&\frac{1}{\left(\kappa_t+\kappa_r\right)}\times\left[tr(\mathbf{\Xi}\mathbf{Y})+\mu||h_{SU}||_2^2\right]
\\s.t.\ &tr\left(\mathbf{E}_n\mathbf{Y}\right)=\mu\alpha^2,\ n=1,2\ldots N
\\&tr\left(\mathbf{E}_{N+1}\mathbf{Y}\right)=\mu
\\&tr\left(\mathbf{\Xi Y}\right)+\mu\left[||h_{SU}||_2^2+\frac{\sigma_w^2}{P\left(\kappa_t+\kappa_r\right)}\right]=1
\end{align}
\end{subequations}
which is currently a SDP problem and can be solved by existing techniques \cite{A.Nemirovski2008(ActaNumerica)}.

{\color{black} However, the matrix $\mathbf{\Xi}$ involves the stochastic phase errors, which are generally unknown due to their randomness and prevent us from predetermining $\mathbf{\Xi}$ and obtaining the solution in reality. In view of this issue, we will further calculate the expectation of $\mathbf{\Xi}$, denoted by $\mathbb{E}_{\mathbf{\Theta}_E}\left[\mathbf{\Xi}\right]$, in order to facilitate the optimization procedure and achieve a statistical average optimization result.}

{\color{black}\subsection{Expectation of $\mathbf{\Xi}$}} 

According to (\ref{eq2-24}), $\mathbb{E}_{\mathbf{\Theta}_E}\left[\mathbf{\Xi}\right]$ can be written as
\begin{equation}\label{eq2-29}
\mathbb{E}_{\mathbf{\Theta}_E}\left[\mathbf{\Xi}\right]=
\mathbb{E}_{\mathbf{v}_E}\left[\mathbf{\Xi}\right]=
\left(\begin{matrix}\mathbf{D}_{IU}\mathbf{D}_{SI}\mathbb{E}_{\mathbf{v}_E}\!\!\left[\mathbf{v}_E\mathbf{v}_E^H\right]\mathbf{D}_{SI}^H\mathbf{D}_{IU}^H&h_{SU}^*\mathbf{D}_{IU}\mathbf{D}_{SI}\mathbb{E}_{\mathbf{v}_E}\!\!\left[\mathbf{v}_E\right]\\\mathbb{E}_{\mathbf{v}_E}\!\!\left[\mathbf{v}_E^H\right]\mathbf{D}_{SI}^H\mathbf{D}_{IU}^Hh_{SU}&0\\\end{matrix}\right)
\end{equation}
where $\mathbf{D}_{SI}=diag(\mathbf{h}_{SI})$ and $\mathbf{v}_E=\left(e^{j\theta_{E1}},e^{j\theta_{E2}},\ldots,e^{j\theta_{EN}}\right)^T$. $\mathbb{E}_{\mathbf{v}_E}\left[\mathbf{v}_E\mathbf{v}_E^H\right]$, which is expressed as
\begin{equation}\label{eq2-30}
\begin{split}
\mathbb{E}_{\mathbf{v}_E}\left[\mathbf{v}_E\mathbf{v}_E^H\right]=
\left(\begin{matrix}\begin{matrix}1&\mathbb{E}_{\delta_{\theta}}\!\!\left[e^{j\theta_{E1}-j\theta_{E2}}\right]\\\mathbb{E}_{\delta_{\theta}}\!\!\left[e^{j\theta_{E2}-j\theta_{E1}}\right]&1\\\end{matrix}&\begin{matrix}\cdots&\mathbb{E}_{\delta_{\theta}}\!\!\left[e^{j\theta_{E1}-j\theta_{EN}}\right]\\\cdots&\mathbb{E}_{\delta_{\theta}}\!\!\left[e^{j\theta_{E2}-j\theta_{EN}}\right]\\\end{matrix}\\\begin{matrix}\vdots&\vdots\\\mathbb{E}_{\delta_{\theta}}\!\!\left[e^{j\theta_{EN}-j\theta_{E1}}\right]&\mathbb{E}_{\delta_{\theta}}\!\!\left[e^{j\theta_{EN}-j\theta_{E2}}\right]\\\end{matrix}&\begin{matrix}\ddots&\ \ \ \ \ \ \ \ \ \vdots\ \ \ \ \ \ \ \ \ \\\cdots&\ \ \ 1\ \ \ \\\end{matrix}\\\end{matrix}\right)
\end{split}
\end{equation}
represents the autocorrelation matrix of $\mathbf{v}_E$, where $\delta_{\theta}=\theta_{Ei}-\theta_{Ej}$ obeys triangular distribution on $[-\pi,\pi]$ as $\theta_{Ei}$ obeys uniform distribution on $\left[-\pi/2,\pi/2\right]$ (detailed in Appendix A). {\color{black} Because $\mathbb{E}_{\delta_{\theta}}\left[e^{j\theta_{Ei}-j\theta_{Ej}}\right]=\mathbb{E}_{\delta_{\theta}}\left[e^{j\delta_{\theta}}\right]=\int_{-\pi}^{\pi}{f\left(\delta_{\theta}\right)e^{j\delta_{\theta}}d\delta_{\theta}}=4/\pi^2$, where $f\left(\delta_{\theta}\right)$, expressed as (\ref{eq-A6}) in Appendix A, is the probability density function of $\delta_{\theta}$, we have
$\mathbb{E}_{\mathbf{v}_E}\left[\mathbf{v}_E\mathbf{v}_E^H\right]=\mathbf{I}_N+\mathbf{J}$, where the $(i,j)$-th element in the matrix $\mathbf{J}$ is expressed as
\begin{equation}
\left[\mathbf{J}\right]_{(i,j)}=\left\{
\begin{matrix}
0,\ \ \ i=j\\
\frac{4}{\pi^2},\ \ i\neq j
\end{matrix}
\right.
\end{equation}

Moreover}, because
$\mathbb{E}_{\mathbf{v}_E}\left[\mathbf{v}_E\right]
=\left(\mathbb{E}_{\theta_{Ei}}\!\!\left[e^{j\theta_{E1}}\right],\mathbb{E}_{\theta_{Ei}}\!\!\left[e^{j\theta_{E2}}\right],\ldots,\mathbb{E}_{\theta_{Ei}}\!\!\left[e^{j\theta_{EN}}\right]\right)^T$
and
$\mathbb{E}_{\theta_{Ei}}\left[e^{j\theta_{Ei}}\right]
=\int_{-\frac{\pi}{2}}^{\frac{\pi}{2}}{f(\theta_{Ei})e^{j\theta_{Ei}}d\theta_{Ei}}
=\int_{-\frac{\pi}{2}}^{\frac{\pi}{2}}{f(\theta_{Ei})(cos\theta_{Ei}+jsin\theta_{Ei})d\theta_{Ei}}
=2/\pi$ for $i=1,2,...,N$,
where $f\left(\theta_{Ei}\right)=1/\pi$ is the probability density function of $\theta_{Ei}$, we have
$\mathbb{E}_{\mathbf{v}_E}\left[\mathbf{v}_E\right]=\left(2/\pi\right)\mathbf{1}^T$.

By substituting {\color{black}$\mathbb{E}_{\mathbf{v}_E}\left[\mathbf{v}_E\mathbf{v}_E^H\right]=\mathbf{I}_N+\mathbf{J}$} and $\mathbb{E}_{\mathbf{v}_E}\left[\mathbf{v}_E\right]=\left(2/\pi\right)\mathbf{1}^T$ into (\ref{eq2-29}), we have
{\color{black}\begin{equation}\label{eq2-35}
\mathbb{E}_{\mathbf{\Theta}_E}\left[\mathbf{\Xi}\right]=\left(\begin{matrix}\mathbf{D}_{IU}\mathbf{D}_{SI}(\mathbf{I}_N+\mathbf{J})\mathbf{D}_{SI}^H\mathbf{D}_{IU}^H&\frac{2}{\pi}h_{SU}^*\mathbf{D}_{IU}\mathbf{D}_{SI}\mathbf{1}^T\\\frac{2}{\pi}\mathbf{1}\mathbf{D}_{SI}^H\mathbf{D}_{IU}^Hh_{SU}&0\\\end{matrix}\right)
\end{equation}

Consequently}, by replacing $\mathbf{\Xi}$ in (P5) with $\mathbb{E}_{\mathbf{\Theta}_E}\left[\mathbf{\Xi}\right]$, we obtain 
\begin{subequations}\label{eq2-44}
\begin{align}
(\mathrm{P6}):\ \mathop{\max}\limits_{\mathbf{Y}\succeq\mathbf{0},\widetilde{\mu}\geq 0}&\frac{1}{\left(\kappa_t+\kappa_r\right)}\times\left[tr(\mathbb{E}_{\mathbf{\Theta}_E}\left[\mathbf{\Xi}\right]\mathbf{Y})+\widetilde{\mu}||h_{SU}||_2^2\right]
\\s.t.\ &tr\left(\mathbf{E}_n\mathbf{Y}\right)=\widetilde{\mu}\alpha^2,\ n=1,2\ldots N
\\&tr\left(\mathbf{E}_{N+1}\mathbf{Y}\right)=\widetilde{\mu}
\\&tr\left(\mathbb{E}_{\mathbf{\Theta}_E}\left[\mathbf{\Xi}\right]\mathbf{Y}\right)+\widetilde{\mu}\left[||h_{SU}||_2^2+\frac{\sigma_w^2}{P\left(\kappa_t+\kappa_r\right)}\right]=1
\end{align}
\end{subequations}
where
$\widetilde{\mu}=\frac{1}{tr(\mathbb{E}_{\mathbf{\Theta}_E}\left[\mathbf{\Xi}\right]\mathbf{X})+||h_{SU}||_2^2+\frac{\sigma_w^2}{P\left(\kappa_t+\kappa_r\right)}}$.

{\color{black}Currently, because the matrix $\mathbb{E}_{\mathbf{\Theta}_E}\left[\mathbf{\Xi}\right]$ in the OBF and constraints only includes the channel coefficients, which can be estimated via existing channel estimation techniques, it is easy for us to configure $\mathbb{E}_{\mathbf{\Theta}_E}\left[\mathbf{\Xi}\right]$, which assists us in completing the phase shift optimization in terms of maximizing the average SNR in the presence of HWI.}

It is remarkable that after (P6) is solved, the $\bm{\theta}^T$ in the $(N+1)$-th row of the $\mathbf{X}=\widetilde{\mu}^{-1}\mathbf{Y}$ in the solution can be extracted to reconstruct $\mathbf{Y}$ based on (\ref{eq2-26}) and $\mathbf{Y}=\widetilde{\mu}\mathbf{X}$. If the reconstructed $\mathbf{Y}$, denoted by $\mathbf{Y}_r$, {satisfies} $\mathbf{Y}_r=\mathbf{Y}$ and $rank(\mathbf{Y}_r)=1$, the $\bm{\theta}^T$ can be regarded as the optimal phase shift vector. As a result, {\color{black}we will test the values in $\mathbf{Y}_r$ and the rank of $\mathbf{Y}_r$ in the simulations in Section VI, in order to investigate whether} the optimal IRS phase shifts can be acquired from $\bm{\theta}^T$ in the $(N+1)$-th row of the $\mathbf{X}=\widetilde{\mu}^{-1}\mathbf{Y}$ in the solution of the relaxed problem.

{\color{black} In addition, after the phase shift optimization process, two possible factors may still remain and influence the performance.
1) Most channel estimation methods suffer from estimation errors, which lead to imperfect CSI of $\mathbf{h}_{IU}$, $\mathbf{h}_{SI}$ and $h_{SU}$. Based on the imperfect CSI, we can only construct an inaccurate $\mathbb{E}_{\mathbf{\Theta}_E}\left[\mathbf{\Xi}\right]$, and then acquire a non-optimal phase shift vector. 
2) Due to the inherent hardware imperfection, synchronization offset and estimation accuracy limit in the real world, the optimized phase values may not be precisely obtained. In this case, we may actually obtain $\widetilde{\bm{\theta}}^T=\bm{\theta}^T\odot\bm{\theta}_p^T$ instead of $\bm{\theta}^T$ after the optimization, where $\bm{\theta}_p=(e^{j\theta_{p1}},e^{j\theta_{p2}},...,e^{j\theta_{pN}})^T$ denotes a residual phase noise vector with $\theta_{pi}$ being the $i$-th random residual phase noise, which may disturb $\bm{\theta}^T$ and affect the optimization performance.  

The performance degradation caused by the aforementioned two factors is worth to be discussed. Thus, we will present the discussions on them in Section VI.}

$\ $

\section{Comparisons with DF Relay}
The DF relay is a conventional active {\color{black}approach} which is also applied for data transmission enhancement in the wireless communication network. Hence, it is {\color{black}necessary} to compare the performance of the IRS with that of the DF relay in the same situation. 
It was already confirmed that the ideal-hardware IRS equipped with a large number of reflecting units could help the wireless communication system provide higher ACR than the ideal-hardware DF relay equipped with one antenna \cite{E.Bjornson2020(WCL)}. {\color{black} However, the comparisons in \cite{E.Bjornson2020(WCL)} were made in consideration of single-antenna DF relay and multiple-unit IRS without HWI. Note that as $N$ grows, the average ACR of the multiple-unit IRS-aided communication system increases, while that of the single-antenna DF relay assisted communication system remains constant under a certain condition. This might be unfair for the DF relay during the comparisons. Therefore, in this section, we will compare the performances of the IRS with $N$ passive reflecting units and the DF relay with $N$ active antennas,} for the purpose of exploring whether the IRS can still possess advantages in ACR {\color{black}and utility} over the {\color{black}multiple-antenna DF relay} when there exists HWI.

{\color{black} Before the comparisons, determining the exact closed-form ACR of the multiple-antenna DF relay assisted communication system with respect to $N$ in the presence of HWI, is a hard nut to crack,}
as the channel coefficients include random phase shifts, which cannot be compensated by the DF relay. {\color{black} Fortunately, we realize that the source-to-relay and the relay-to-destination channels can similarly be regarded as the uplink and downlink channels modelled in \cite{E.Bjornson2014(TIT)}, which assists us in establishing the \textit{closed-form upper bound} of the ACR in relation to $N$ for the multiple-antenna DF relay supported communication system. }

Let {\color{black}$\mathbf{h}_{SR}$, $\mathbf{h}_{RU}$} and $h_{SU}$ denote the source-to-relay, relay-to-destination and source-to-destination channels, respectively. 
For comparing, we assume that {\color{black}$\mathbf{h}_{SR}=\mathbf{h}_{SI}$, $\mathbf{h}_{RU}=\mathbf{h}_{IU}$}, and the receiver noises at the DF relay and the destination terminal have the same variance of $\sigma_w^2$. If the HWI appears at the source transmitter, the DF relay and the destination receiver, {\color{black}according to Eq. (6) and Eq. (2) in \cite{E.Bjornson2014(TIT)},} the signals received by the DF relay and the destination terminal are modelled as
{\color{black}\begin{equation}\label{eq2-47}
\mathbf{y}_{DF}(t)=\mathbf{h}_{SR}\left[\sqrt{P_1}s(t)+\eta_t(t)\right]+\bm{\eta}_{r_{DF}}(t)+\mathbf{w}_{DF}(t)
\end{equation}
and
\begin{equation}\label{eq2-48}
y_{U1}(t)=\mathbf{h}_{RU}\left[\sqrt{P_2}\mathbf{s}(t)+\bm{\eta}_{t_{DF}}(t)\right]+\eta_{r1}(t)+w(t)
\end{equation}
\begin{equation}\label{eq2-49}
y_{U2}(t)=h_{SU}\left[\sqrt{P_1}s(t)+\eta_t(t)\right]+\eta_{r2}(t)+w(t)
\end{equation}
where} $P_1$ and $P_2$ are the transmitting powers of the source and the DF relay under the constraint of $P=\frac{P_1+P_2}{2}$ \cite{E.Bjornson2020(WCL)}; $y_{U1}(t)$ and $y_{U2}(t)$ are the signals received by the destination terminal through channel {\color{black}$\mathbf{h}_{RU}$} and $h_{SU}$, respectively; $\mathbf{w}_{DF}(t)\sim\mathcal{CN}(\mathbf{0},\sigma_w^2\mathbf{I})$ and $w(t)\sim\mathcal{CN}(0,\sigma_w^2)$ are the receiver noises at the DF relay and the destination terminal;
{\color{black} $\eta_t(t)\sim\mathcal{CN}(0,\Upsilon_t)$, $\bm{\eta}_{r_{DF}}(t)\sim\mathcal{CN}(\mathbf{0},\mathbf{V}_{r_{DF}})$, $\bm{\eta}_{t_{DF}}(t)\sim\mathcal{CN}(\mathbf{0},\mathbf{\Upsilon}_{t_{DF}})$, $\eta_{r1}(t)\sim\mathcal{CN}(0,V_{r1})$ and $\eta_{r2}(t)\sim\mathcal{CN}(0,V_{r2})$} are the distortion noises at the source transmitter, the DF-relay receiver, the DF-relay transmitter and the destination receiver, respectively, with {\color{black} $\Upsilon_t$, $\mathbf{V}_{r_{DF}}$, $\mathbf{\Upsilon}_{t_{DF}}$, $V_{r1}$ and $V_{r2}$ given by
\begin{equation}
\Upsilon_t=\kappa_tP_1\mathbb{E}[s(t)s^*(t)]
\end{equation}
\begin{equation}
\mathbf{V}_{r_{DF}}=\kappa_{r_{DF}}P_1\mathbb{E}[s(t)s^*(t)]\times diag(|h_{SR,1}|^2,...,|h_{SR,N}|^2)
\end{equation}
\begin{equation}
\mathbf{\Upsilon}_{t_{DF}}=\kappa_{t_{DF}}P_2\times diag\{\mathbb{E}[s_1(t)s_1^*(t)],...,\mathbb{E}[s_N(t)s_N^*(t)]\}
\end{equation}
\begin{equation}
V_{r1}=\kappa_{r1}P_2\mathbf{h}_{RU}^T\mathbb{E}[\mathbf{s}(t)\mathbf{s}^H(t)]\mathbf{h}_{RU}^*
\end{equation}
\begin{equation}
V_{r2}=\kappa_{r2}P_1 |h_{SU}|^2\mathbb{E}[s(t)s^*(t)]
\end{equation}
where $\mathbf{s}(t)$ denotes the signal transmitted by the DF relay at time $t$; $s_i(t)$, for $i=1,2,...,N$, represents the $i$-th transmit symbol in $\mathbf{s}(t)$, with $\mathbb{E}[s_i(t)s_i^*(t)]=1$;} $\kappa_t$, $\kappa_{r_{DF}}$, $\kappa_{t_{DF}}$, $\kappa_{r1}$ and $\kappa_{r2}$ are the proportionality factors.

For simple {\color{black}analysis}, we consider that
$\kappa_{t_{DF}}=\kappa_t$ and $\kappa_{r1}=\kappa_{r2}=\kappa_{r_{DF}}=\kappa_r$,
as the hardware characteristics of the transceivers in the DF relay are similar to those in the source equipment and the destination terminal. Therefore, {\color{black}referring to Eq. (26), Eq. (27) in \cite{E.Bjornson2014(TIT)}, and Eq. (15) in \cite{J.N.Laneman2004(TIT)},} the {\color{black}upper bound} of the ACR of the multiple-antenna DF relay assisted communication system with HWI is expressed as 
\begin{equation}\label{eq2-55}
R_{HWI}^{DF}(N)=
\frac{1}{2}\min{\left\{\mathfrak{A}(N),\mathfrak{B}(N)\right\}}
\end{equation}
where 
{\color{black}\begin{equation}\label{eq-A}
\mathfrak{A}(N)=\log_2\left(1+\frac{N\mu_{SI}}{\kappa_r\mu_{SI}+N\kappa_t\mu_{SI}+\frac{\sigma_w^2}{P_1} }\right)
\end{equation}
\begin{equation}\label{eq-B}
\mathfrak{B}(N)=\log_2\left(1+\frac{\mu_{SU}}{(\kappa_t+\kappa_r)\mu_{SU}+\frac{\sigma_w^2}{P_1}}+
\frac{N\mu_{IU}}{\kappa_t\mu_{IU}+N\kappa_r\mu_{IU}+\frac{\sigma_w^2}{P_2} }\right)
\end{equation}

Correspondingly, the utility of the multiple-antenna DF relay is expressed as
\begin{equation}\label{eq-DF Utility}
\gamma_{HWI}^{DF}(N)=
\left\{\begin{matrix}
\frac{\kappa_r\mu_{SI}^2+\frac{\sigma_w^2}{P_1}\mu_{SI}}
{2\left(\kappa_r\mu_{SI}+N\kappa_t\mu_{SI}+\frac{\sigma_w^2}{P_1}\right)\left(\kappa_r\mu_{SI}+N\kappa_t\mu_{SI}+\frac{\sigma_w^2}{P_1}+N\mu_{SI}\right)\ln2}, \ \ \mathfrak{A}(N)<\mathfrak{B}(N) \\
\frac{\kappa_t\mu_{IU}^2+\frac{\sigma_w^2}{P_2}\mu_{IU}}
{2\left(1+\frac{\mu_{SU}}{(\kappa_t+\kappa_r)\mu_{SU}+\frac{\sigma_w^2}{P_1}}+
\frac{N\mu_{IU}}{\kappa_t\mu_{IU}+N\kappa_r\mu_{IU}+\frac{\sigma_w^2}{P_2}}\right)
\left(\kappa_t\mu_{IU}+N\kappa_r\mu_{IU}+\frac{\sigma_w^2}{P_2}\right)^2 \ln2}, \ \ \mathfrak{A}(N)>\mathfrak{B}(N)
\end{matrix}\right.
\end{equation}

For analysis convenience and symbol unification, we assume that $\kappa_t+\kappa_r=\kappa$ with $\kappa_t=\kappa_r=\frac{1}{2}\kappa$ \cite{E.Bjornson2014(TIT)}, and} the total transmitting power of the DF relay assisted communication system ($P_1+P_2=2P$) is allocated by $P_1=P_2=P$. {\color{black} Subsequently, in order to investigate whether the IRS is potentially capable of outperforming the DF relay in the presence of HWI, we will compare $R_{HWI}^{DF}(N)$ in (\ref{eq2-55}) with $\overline{R_{HWI}}(N)$ in (\ref{eq2-17}) from the perspective of scaling law, by considering first $N\rightarrow\infty$ and then $P\rightarrow\infty$ in the following \textbf{Lemma 2} and  \textbf{Lemma 3}.
}

{\color{black}
\begin{lemma}
When $N\rightarrow\infty$, we have
\begin{equation}\label{N_inf_ACR_Compare}
\left.\overline{R_{HWI}}(N)\right|_{N\rightarrow\infty}>\left.R_{HWI}^{DF}(N)\right|_{N\rightarrow\infty}
\end{equation}
\begin{equation}\label{N_inf_Utility_Compare}
\left.\gamma_{HWI}(N)\right|_{N\rightarrow\infty}=\left.\gamma_{HWI}^{DF}(N)\right|_{N\rightarrow\infty}=0
\end{equation}
\end{lemma}

\begin{proof}
The proof is given in Appendix B.
\end{proof}

\begin{lemma}
When $P\rightarrow\infty$, we have
\begin{equation}\label{P_inf_ACR_Compare}
\left.\overline{R_{HWI}}(N)\right|_{P\rightarrow\infty}>\left.R_{HWI}^{DF}(N)\right|_{P\rightarrow\infty}
\end{equation}
\begin{equation}\label{P_inf_Utility_Compare_IRS}
\left.\gamma_{HWI}(N)\right|_{P\rightarrow\infty}=0
\end{equation}
\begin{equation}\label{P_inf_Utility_Compare_DF}
\left.\gamma_{HWI}^{DF}(N)\right|_{P\rightarrow\infty}
=\frac{\kappa}{(\kappa+N\kappa+2N)(\kappa+N\kappa)\ln 2}
\end{equation}
\end{lemma}

\begin{proof}
The proof is given in Appendix C.
\end{proof}

\textbf{Lemma 2} and  \textbf{Lemma 3} demonstrate that: 1) when $N$ becomes large enough, or when $P$ is sufficiently high, the IRS can surpass the conventional multiple-antenna DF relay in terms of the ACR performance in the presence of HWI. 2) If $N\rightarrow\infty$, both the utilities of the IRS and the multiple-antenna DF relay verge on zero, hinting that adding one more reflecting element on the IRS or one more antenna on the DF relay hardly improves the ACR. 3) When $P$ is nearly infinite, the utility of the IRS is close to zero, while that of the multiple-antenna DF relay converges to a positive value, which indicates that adding one more antenna on the DF relay can still improve the ACR. This is because the IRS is passive and almost useless when $P\rightarrow\infty$, which makes the line-of-sight (LoS) link infinitely strong, while the DF relay is active and consumes power when retransmitting the wireless signals, so that when $P\rightarrow\infty$, each active antenna can always possess positive transmitting power and contribute to the data transmission enhancement. On this point, the multiple-antenna DF relay is more advantageous.
}

Moreover, it can be predicted that the IRS may possibly always outperform the multiple-antenna DF relay when the level of the transceiver HWI is high, because the HWI at the IRS is modelled as a phase error matrix which does not contain $\kappa_t$ or $\kappa_r$, while the HWI at the DF relay involves the distortion noises which contain the two terms. The DF relay may perform worse with higher $\kappa_t+\kappa_r$ while the IRS may maintain the performance due to the fixed uniform distribution of the phase errors. {\color{black}Therefore, we will also derive the interval of $\kappa_t+\kappa_r$, in which the IRS can always surpass the DF relay for all $N>0$ in the following \textbf{Lemma 4}.
}

\begin{lemma}
The IRS will always outperform the {\color{black}multiple-antenna} DF relay for all $N>0$, when $\kappa_t+\kappa_r$ satisfies
{\color{black}\begin{equation}\label{eq2-69}
\kappa_t+\kappa_r>2\sigma_w^4 \left[P^2(\beta+\lambda+\mu_{SU})^2-2\sigma_w^2P(\beta+\lambda+\mu_{SU})\right]^{-1}
=\kappa_{th}
\end{equation}
where $\beta$ and $\lambda$ have been defined in (\ref{eq-revision-1}) and (\ref{eq-revision-2}), respectively.}
\end{lemma}
\begin{proof}
The proof is given in Appendix D.
\end{proof}

\textbf{Lemma 4} demonstrates that $\kappa_t+\kappa_r$ determines whether the IRS can always outperform the DF relay for all $N>0$ by a threshold $\kappa_{th}$ in (\ref{eq2-69}), {\color{black}which is mainly} decided by $\mu_{SI}$, $\mu_{IU}$, $\mu_{SU}$, $P$ and $\sigma_w^2$. {\color{black} If $P\rightarrow\infty$, we have $\kappa_{th}\rightarrow 0$, which makes (\ref{eq2-69}) always hold and makes the IRS perform better for all $N>0$ and $\kappa_t+\kappa_r>0$ in terms of the ACR. The outcome is consistent with (\ref{P_inf_ACR_Compare}) in \textbf{Lemma 3}.}

\section{Simulation Results}
\subsection{System Setup and Parameter Setting}
This section will numerically elaborate the results of the ACR {\color{black} and the IRS utility} with or without HWI, and compare the performances of the IRS and the DF relay. 
As shown in Figure \ref{Fig-coordinate}, a two-dimensional plane in meters is established to indicate the positions of the source, the IRS and the destination, 
\begin{figure}[!t]
\includegraphics[width=3.6in]{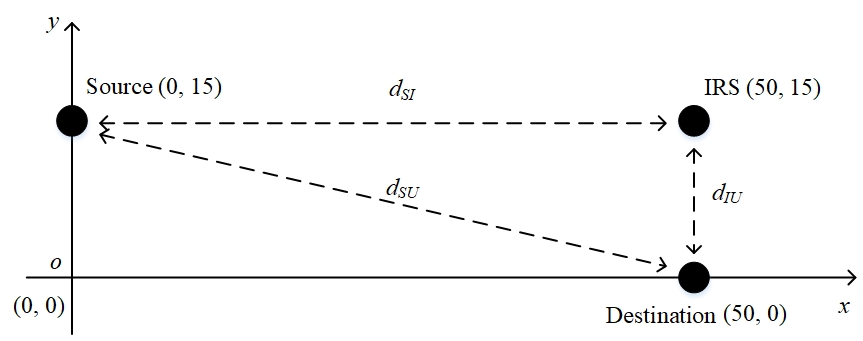}
\centering
\caption{Communication system design in the simulations. Three dashed lines which indicate $d_{SI}$, $d_{IU}$ and $d_{SU}$ constitute a right triangle, where $d_{SU}=\sqrt{d_{SI}^2+d_{IU}^2}$.}
\label{Fig-coordinate}
\end{figure}
which are placed at $(0,15)$, $(50,15)$ and $(50,0)$. Regardless of the height, the distances between the source and the IRS ($d_{SI}$), the IRS and the destination ($d_{IU}$) and the source and the destination ($d_{SU}$) are $d_{SI}=50$ $m$, $d_{IU}=15$ $m$ and $d_{SU}=\sqrt{d_{SI}^2+d_{IU}^2}\approx52.2$ $m$, respectively. According to \cite{M.Cui2019(WCL),E.Bjornson2014(TIT)}, the other parameters are set {\color{black}in Table I. Based on Table I}, the power attenuation coefficients of channel $\mathbf{h}_{IU}$ (or $\mathbf{h}_{RU}$), $\mathbf{h}_{SI}$ (or $\mathbf{h}_{SR}$) and $h_{SU}$ are derived by $\sqrt{\mu_{IU}}=\sqrt{\zeta_0(d_0/d_{IU})^{\alpha_{IU}}}$, $\sqrt{\mu_{SI}}=\sqrt{\zeta_0(d_0/d_{SI})^{\alpha_{SI}}}$ and $\sqrt{\mu_{SU}}=\sqrt{\zeta_0(d_0/d_{SU})^{\alpha_{SU}}}$ \cite{M.Cui2019(WCL)}.

\begin{table}
\renewcommand{\arraystretch}{1.3} 
\caption{{\color{black}Parameter configurations.}}
\label{Table_Parameters}
\centering
{\color{black}
\begin{small} 
\begin{tabular}{ccc}
\hline
Parameters & Definitions & Values\\
\hline
Amplitude Reflection Coefficient & $\alpha$ & $1$\\
Signal Power & $P$ & $20$ dBm\\
Receiver Noise Power & $\sigma_w^2$ & $-80$ dBm\\
Path Loss & $\zeta_0$ & $-20$ dB\\
Reference Distance & $d_0$ & $1$ m\\
Path Loss Exponents & $\alpha_{IU}=\alpha_{SI}=\alpha_{SU}$ & $3$\\
Phase Shift in $\mathbf{h}_{IU}$ & $\varphi_{IU,i}$ & Random in $[0,2\pi]$\\
Phase Shift in $\mathbf{h}_{SI}$ & $\varphi_{SI,i}$ & Random in $[0,2\pi]$\\
Phase Shift in $h_{SU}$ & $\varphi_{SU}$ & $\frac{\pi}{4}$\\
Proportionality Coefficients of Distortion Noises & $\kappa_t=\kappa_r$ & $0.05^2$\\
Oscillator Quality & $\delta$ & $1.58\times 10^{-4}$\\
\hline
\end{tabular}
\end{small}
}
\end{table} 
During the comparisons with DF relay, $d_{SI}$, $d_{IU}$ and $d_{SU}$ are also regarded as the distances between the source and the DF relay, the DF relay and the destination, and the source and the destination, respectively, which still adhere to $d_{SU}=\sqrt{d_{SI}^2+d_{IU}^2}$. The proportionality coefficients can be changed for diverse observations, but still satisfy $\kappa_t=\kappa_r$.

{\color{black}\subsection{Numerical Illustrations for \textbf{Theorem 1} and \textbf{Lemma 1}}}
{\color{black} For further discussing and validating the theoretical analysis in Section III}, we carry out the simulations via the following steps:

\textit{B-Step 1}: We calculate $\overline{R_{HWI}}(N)$ in (\ref{eq2-17}) {\color{black}and $\gamma_{HWI}(N)$ in (\ref{Utility Expression})}, and record the results with HWI from {\color{black} $N=1$ to $N=5000$}.

\textit{B-Step 2}: We calculate $R(N)$ in (\ref{eq2-11}) {\color{black}and $\gamma(N)$ in (\ref{Utility Expression 2})}, and record the results without HWI from {\color{black} $N=1$ to $N=5000$}.

\textit{B-Step 3}: {\color{black}We calculate the rate gap $\delta_R(N)$ in (\ref{eq2-18}) and the utility gap $\delta_{\gamma}(N)$ in (\ref{Utility Degradation}), and record the results from $N=1$ to $N=5000$.}

\textit{B-Step 4}: We calculate and record the numerical results of $R_{HWI}(N)$ in (\ref{eq2-16}) from {\color{black} $N=1$ to $N=5000$}. Due to the randomness of the phase errors generated by the IRS, the ACR is averaged on 1000 Monte Carlo trials {\color{black}every 500 points}.

The average ACRs and {\color{black}IRS utilities} as functions of $N$ from $N=1$ to $N=5000$ are described in Figure \ref{Fig-For_IRS_ACR_and_Utility}. It is indicated that: 1) the experimental results fit well with the theoretical ones from $N=1$ to $N=5000$, which verifies the tightness of (\ref{eq2-17}). 2) The average ACR with HWI is lower and increases more slowly than that without HWI, and the rate gap {\color{black}widens} as $N$ grows. {\color{black}This phenomenon implies that when $N$ grows, the HWI accumulates and begets more severe ACR degradation. 3) When $N$ becomes pretty large, the ACR with HWI verges on $\log_2\left(1+\frac{1}{\kappa_t+\kappa_r}\right)=7.6511$, which testifies the correctness of (\ref{R_HWI Upper Bound}). 4) The IRS utility with HWI is lower than that without HWI, which demonstrates that the HWI reduces the IRS utility as well. Besides, both the IRS utility and the utility gap descend as $N$ grows, which reveals that the influence of the HWI on the IRS utility becomes slighter when $N$ is larger.}

\begin{figure*}[!t]
\centering
\subfloat[]{\includegraphics[width=3.2in]{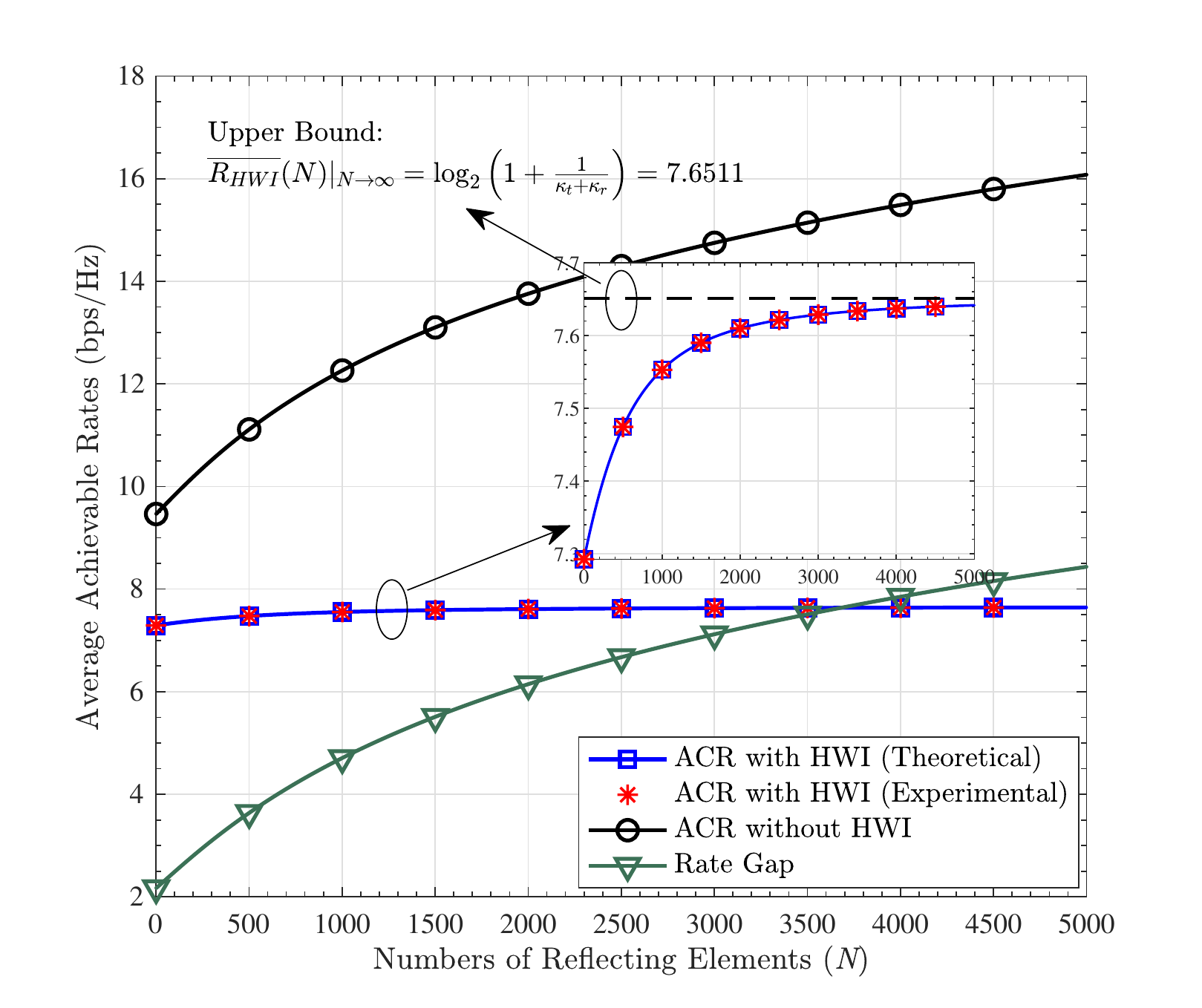}}
\label{For_Average_ACR}
\subfloat[]{\includegraphics[width=3.2in]{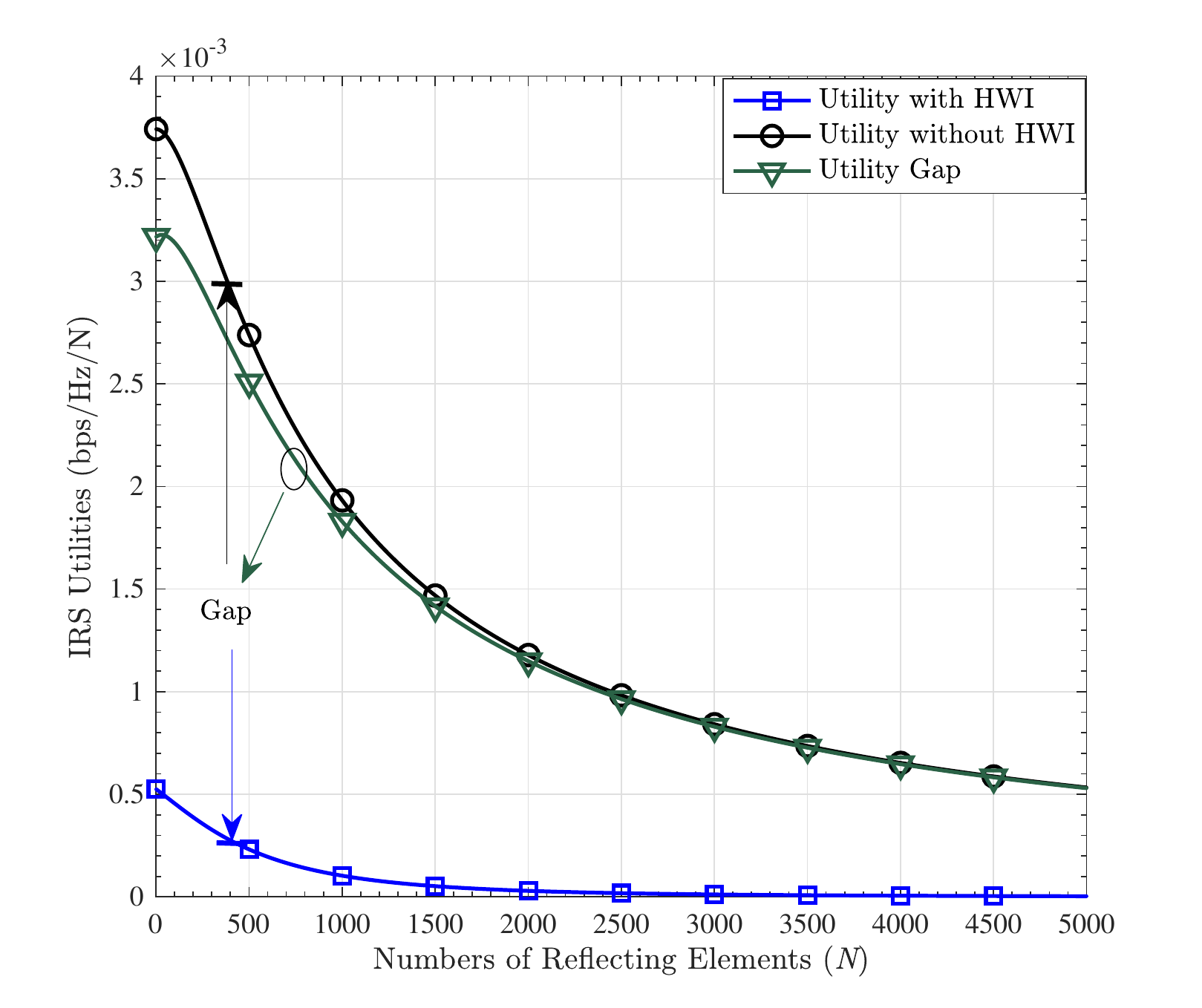}}
\label{For_Utility}
\hfil
\caption{{\color{black} Average ACRs and IRS utilities as functions of $N$ with or without HWI. (a) Average ACRs with respect to $N$, the curves marked with "$\square$", "$\bigcirc$", "$\bigtriangledown$" and "$*$", represent the results obtained in \textit{B-Step 1} to \textit{B-Step 4}, respectively. (b) IRS utilities with respect to $N$, the curves marked with "$\square$", "$\bigcirc$" and "$\bigtriangledown$", represent the results obtained in \textit{B-Step 1} to \textit{B-Step 3}, respectively.}}
\label{Fig-For_IRS_ACR_and_Utility}
\end{figure*}

$\ $

{\color{black} \subsection{Phase Shift Optimization}}
{\color{black}For giving insights into the phase shift optimization approach in Section IV, we carry out the simulations through the following steps:

\textit{C-Step 1}: We solve (P6) by adopting CVX Toolbox with SDPT3 Solver, and obtain the maximum average SNR from the solution of the OBF in (\ref{eq2-44}a). Based on this solution, we calculate and record the ACRs at $N=1,13,25,37$.

\textit{C-Step 2}: We solve (P6) and obtain the optimized matrix $\mathbf{Y}$ and variable $\widetilde{\mu}$. Next, we extract the $\bm{\theta}^T$ in the $(N+1)$-th row of $\mathbf{X}=\widetilde{\mu}^{-1}\mathbf{Y}$. Then, we utilize $\bm{\theta}^T$ to reconstruct $\mathbf{X}$ according to (\ref{eq2-26}) and $\mathbf{Y}$ according to $\mathbf{Y}=\widetilde{\mu}\mathbf{X}$, and denote the reconstructed $\mathbf{X}$ and $\mathbf{Y}$ by $\mathbf{X}_r$ and $\mathbf{Y}_r$, respectively. Finally, we substitute $\mathbf{Y}_r$ into the OBF in (\ref{eq2-44}a) and obtain the average SNR, based on which we calculate and record the ACRs at $N=1,13,25,37$.

\textit{C-Step 3}: Based on the extracted $\bm{\theta}^T$, we obtain the optimized IRS phase shift matrix $\mathbf{\Phi}$ according to $\mathbf{\Phi}=diag(\bm{\theta}^T)$. Then, we substitute $\mathbf{\Phi}$ into (\ref{original ACR with HWI}) and obtain the ACRs with HWI, which are averaged on 1000 Monte Carlo trials at $N=1,13,25,37$.

\begin{figure}[!t]
\includegraphics[width=3.2in]{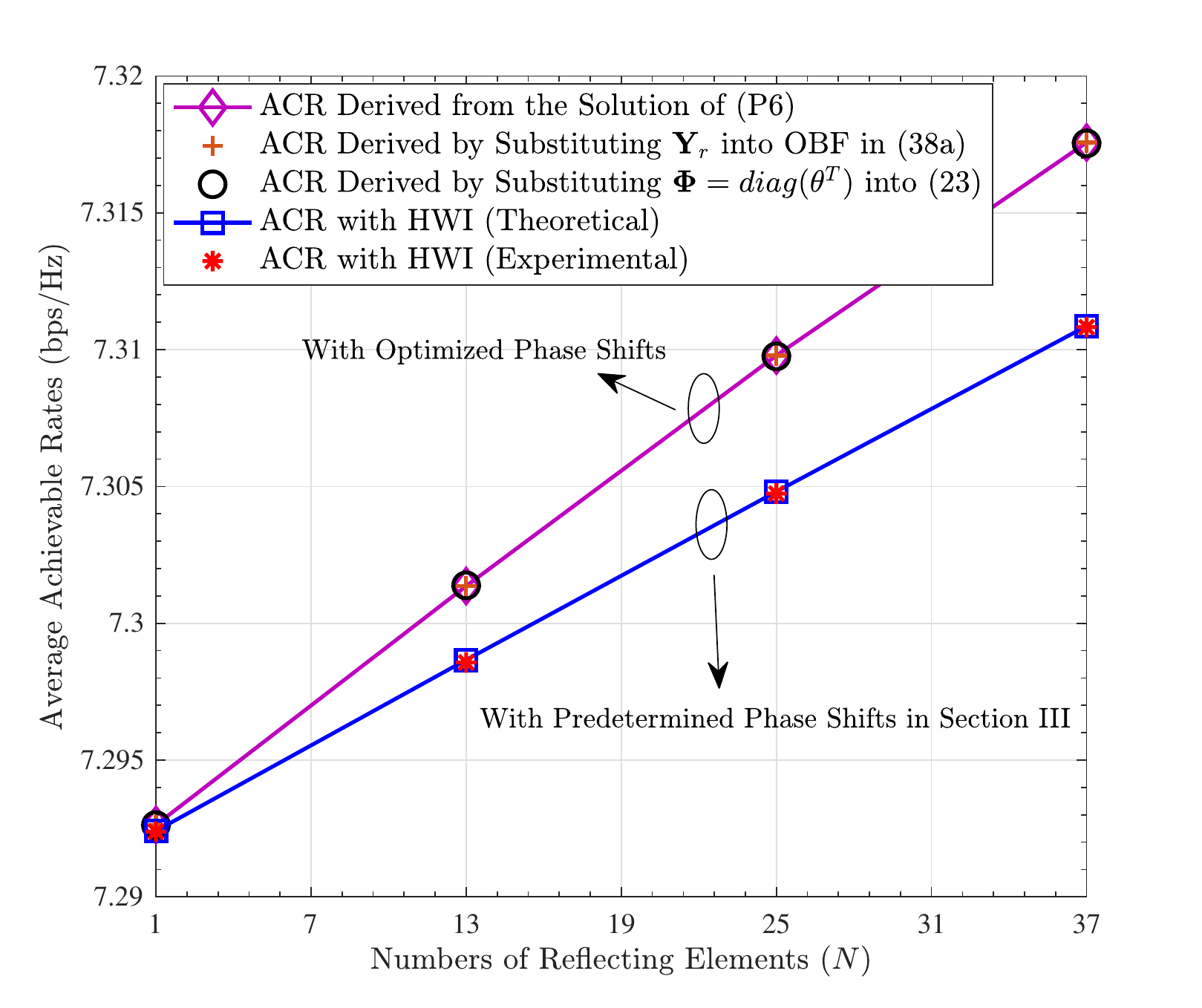}
\centering
\caption{{\color{black}Average ACRs as functions of $N$ with HWI. The curves marked with "$\Diamond$", "+" and "$\bigcirc$" represent the results obtained in \textit{C-Step 1} to \textit{C-Step 3}, respectively. The curves marked with "$\square$" and "$*$" are copied from Figure \ref{Fig-For_IRS_ACR_and_Utility} (a) for comparisons.}}
\label{Fig-Optimization}
\end{figure}

The average ACRs as functions of $N$ with HWI are depicted in Figure \ref{Fig-Optimization}.
Results in Figure \ref{Fig-Optimization} show that:
1) the curves obtained in \textit{C-Step 1} and \textit{C-Step 2} coincide, indicating that $\mathbf{Y}_r=\mathbf{Y}$. Moreover, we calculate the rank of $\mathbf{Y}_r$ and obtain $rank(\mathbf{Y}_r)=1$. Because $\mathbf{Y}_r$ is constructed by $\bm{\theta}^T$ in the $(N+1)$-th row of $\mathbf{X}=\widetilde{\mu}^{-1}\mathbf{Y}$ in the solution, $\bm{\theta}^T$ is testified to be the optimal IRS phase shift vector.
2) The curves obtained in \textit{C-Step 1} and \textit{C-Step 3} coincide, confirming that the mathematical derivations for $\mathbb{E}_{\mathbf{\Theta}_E}\left[\mathbf{\Xi}\right]$ in (\ref{eq2-35}) are correct.
3) The average ACRs with the optimized IRS phase shifts exceed the average ACRs with $\theta_i=-\left(\varphi_{IU,i}+\varphi_{SI,i}\right)$, demonstrating that $\theta_i=-\left(\varphi_{IU,i}+\varphi_{SI,i}\right)$ is not the optimal phase shift as it does not take $h_{SU}$ into account.}

$\ $

{\color{black} \subsection{Discussions on Channel Estimation Errors and Residual Phase Noises}}
{\color{black}Because most IRS-aided communication systems suffer from channel estimation errors, and the optimized IRS phase shifts may generally be affected by residual phase noises, as narrated at the end of Section IV, we will probe into the influence of the two factors on the optimization performance.
The channel estimation errors are set to be additive complex variables according to Eq. (2) in \cite{J.Zhang2020(CL)}, which follow the zero-mean complex Gaussian distribution with the variance of $\sigma_w^2$. More detailed information about the CSI uncertainty models and simulation parameters can be found in \cite{J.Zhang2020(CL)}. The residual phase noises $\theta_{pi}$ in $\bm{\theta}_p$, for $i=1,2,...,N$, are also set to be uniformly distributed on $[-\pi/2,\pi/2]$.

\begin{figure}[!t]
\includegraphics[width=3.2in]{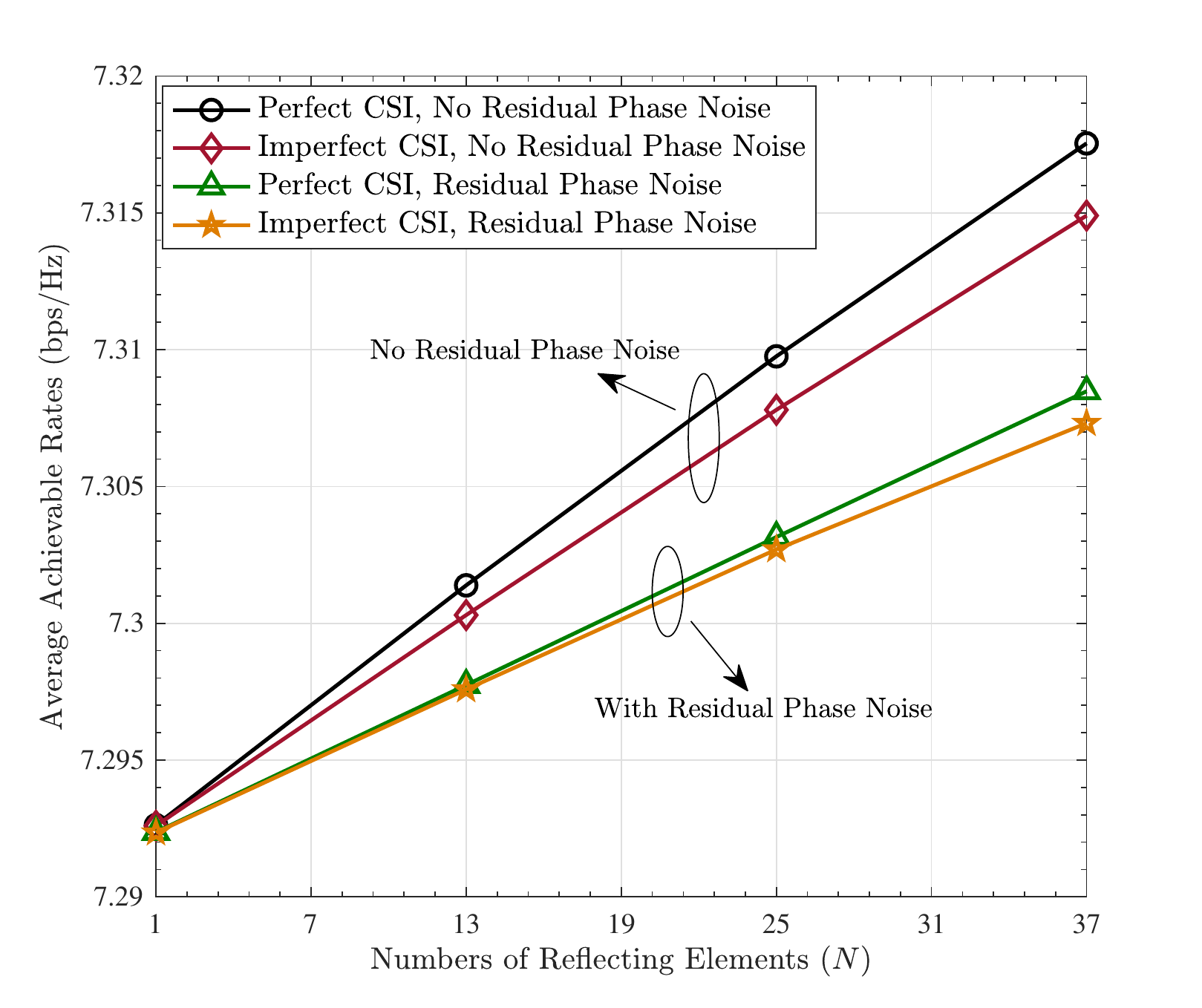}
\centering
\caption{{\color{black} Influences of the channel estimation errors and residual phase noises on the optimization results. The ACRs are derived by substituting $\mathbf{\Phi}=diag(\bm{\theta}^T)$ or $\mathbf{\Phi}=diag(\bm{\theta}^T\odot\bm{\theta}_p^T)$ into (\ref{original ACR with HWI}) and are averaged on 1000 Monte Carlo trials. "Imperfect CSI" means that there are channel estimation errors, while "Perfect CSI" represents the opposite.}}
\label{Fig-For_Optimization_Other_Source_HWI}
\end{figure}

In the simulations, for investigating the average ACR with channel estimation errors, we first adopt the CSI with errors to construct $\mathbb{E}_{\mathbf{\Theta}_E}\left[\mathbf{\Xi}\right]$ and solve (P6), and then substitute $\mathbf{\Phi}=diag(\bm{\theta}^T)$ into (\ref{original ACR with HWI}) which contains the actual CSI. For investigating the average ACR with residual phase noises, we first solve (P6) and exert the influence of $\bm{\theta}_p$ on $\bm{\theta}^T$ by constructing $\bm{\theta}^T\odot\bm{\theta}_p^T$, and then substitute $\mathbf{\Phi}=diag(\bm{\theta}^T\odot\bm{\theta}_p^T)$ into (\ref{original ACR with HWI}).
Figure \ref{Fig-For_Optimization_Other_Source_HWI} depicts the influences of the channel estimation errors and residual phase noises on the optimization results. It is demonstrated that: 1) both the channel estimation errors and the residual phase noises reduce the average ACR and degrade the optimization performance. 2) The residual phase noises impose more serious negative impact on the performance than the channel estimation errors, manifesting that the inherent hardware imperfection, synchronization offset and estimation accuracy limit in the real world, are key potential factors that affect the optimization performance.
}

$\ $

\subsection{Comparisons with DF Relay}

In order to {\color{black}validate} the theoretical {\color{black}analysis} in Section V, we {\color{black}will} numerically compare the ACRs {\color{black}and the utilities} for the IRS-aided and the conventional multiple-antenna DF relay assisted wireless communication systems in the presence of HWI. 
{\color{black} Following Section V, we will compare the performances by varying $N$ and $P$.

$\ $

\textit{1) Comparisons by varying $N$}:

\begin{figure*}[!t]
\centering
\subfloat[]{\includegraphics[width=3.2in]{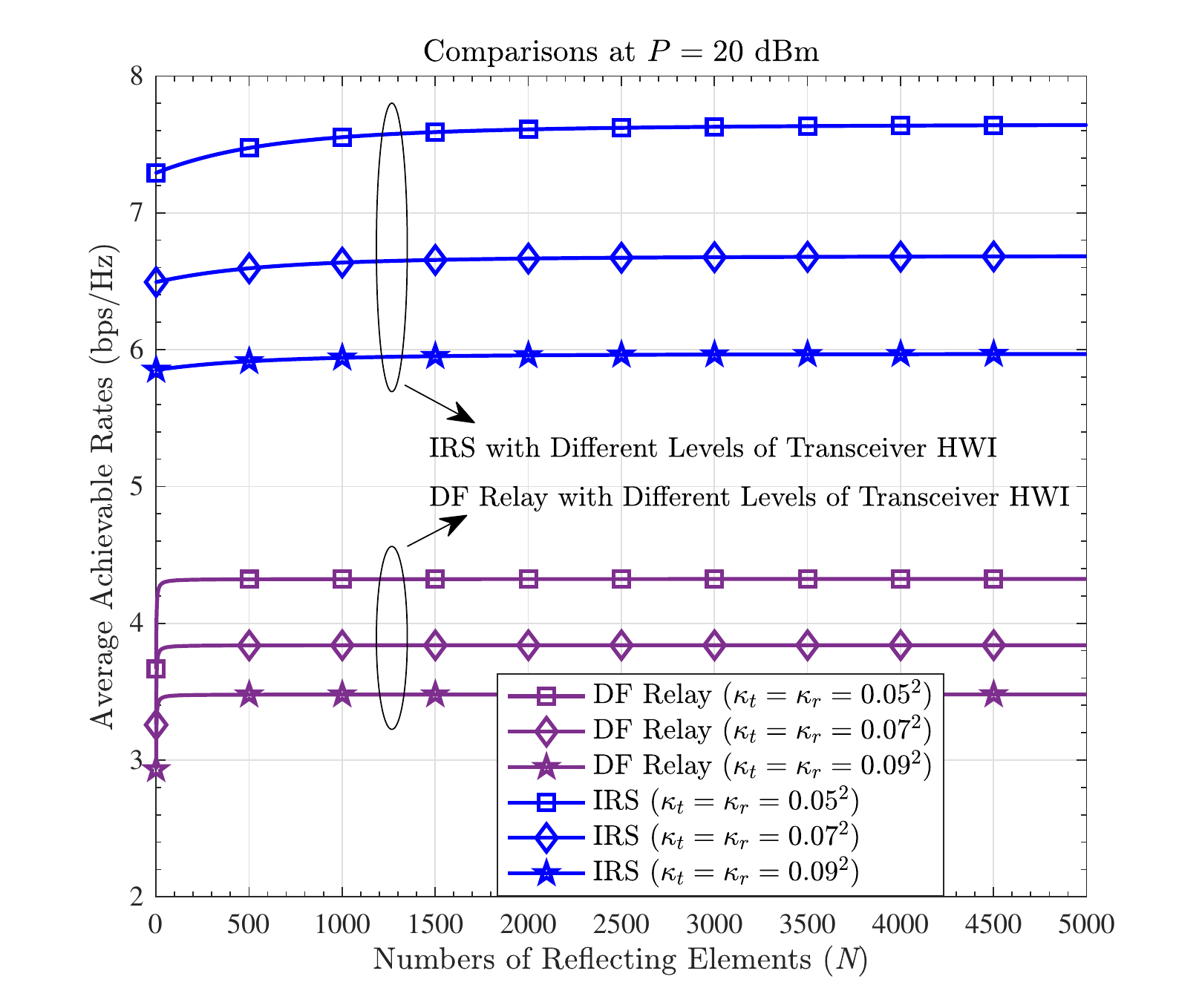}}
\label{Fig-For_N_ACR_Comparison_DF}
\subfloat[]{\includegraphics[width=3.2in]{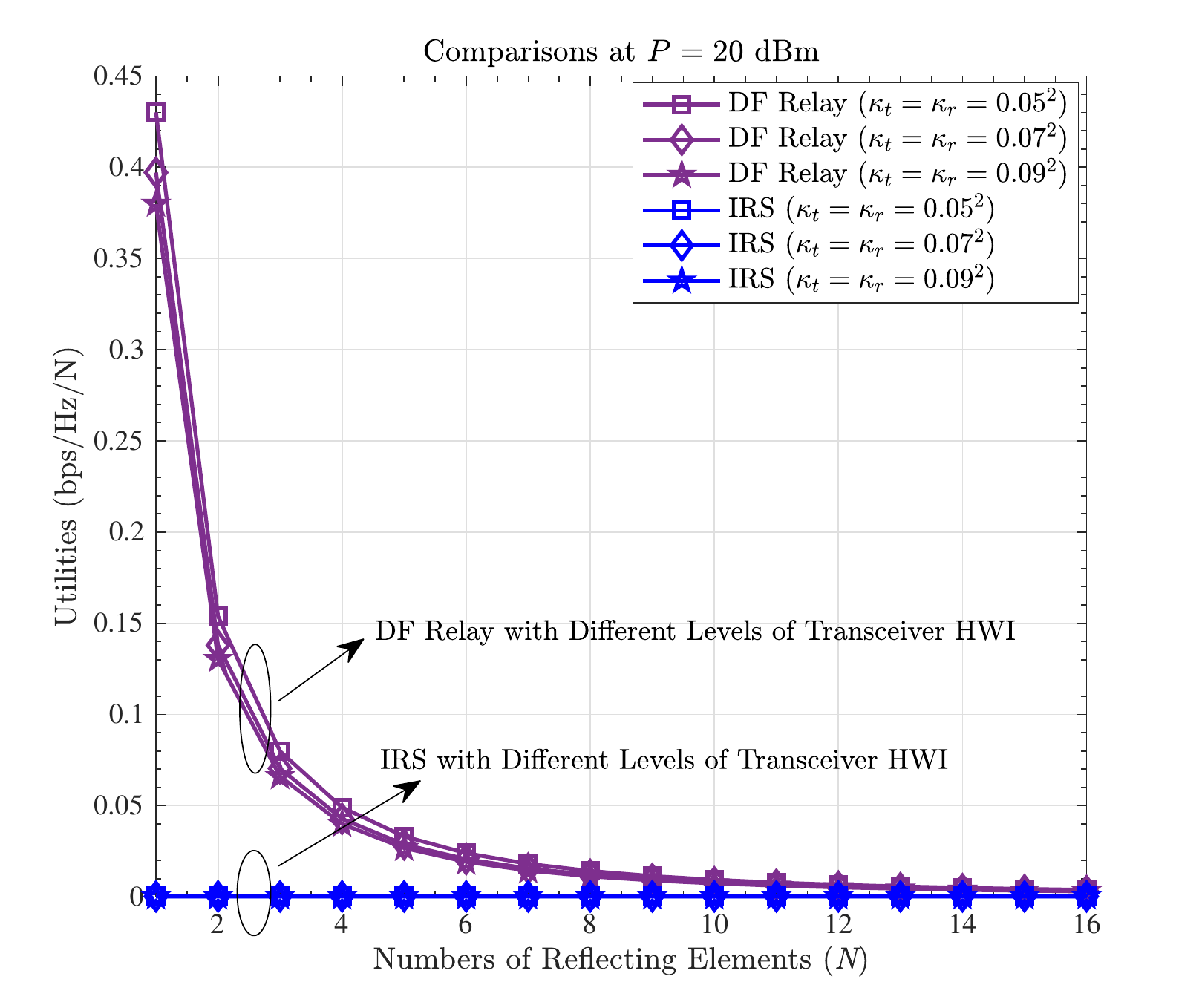}}
\label{Fig-For_N_Utility_Comparison_DF}
\hfil
\caption{{\color{black}Comparisons with DF relay by varying $N$. (a) Average ACRs as functions of $N$, with $P=20$ dBm and $\kappa_t=\kappa_r=0.05^2,0.07^2,0.09^2$. (b) Utilities as functions of $N$, with $P=20$ dBm and $\kappa_t=\kappa_r=0.05^2,0.07^2,0.09^2$.}}
\label{Fig-C1}
\end{figure*}

First, considering the transmitting power to be fixed ($P=20$ dBm), we compare the average ACR in (\ref{eq2-17}) with that in (\ref{eq2-55}) by varying $N$, and observe the simulation results at $\kappa_t=\kappa_r=0.05^2,0.07^2,0.09^2$. Figure \ref{Fig-C1} (a) displays the average ACRs of the IRS-aided and the DF relay assisted wireless communication systems in relation to $N$, from $N=1$ to $N=5000$. It is indicated that: 1) the ACRs decrease when $\kappa_t$ and $\kappa_r$ grow, which verifies that more severe HWI is concomitant with more serious ACR reduction. 2) Although it might not be realistic for the IRS and the DF relay to be equipped with such a large number (e.g. 5000) of reflecting elements and antennas in practical implementations, the result testifies (\ref{N_inf_ACR_Compare}) in \textbf{Lemma 2} and confirms the possibility that when $N$ is extremely large, the IRS is capable of outperforming the DF relay in terms of the ACR performance. 
It is worth noting that when $\kappa_t=\kappa_r=0.05^2,0.07^2,0.09^2$, the IRS always performs better for all $N\in[1,5000]$. This is because with the system parameters set in Section VI-A, the $\kappa_{th}$ in (\ref{eq2-69}) in \textbf{Lemma 4} is computed as $\kappa_{th}=4.0451\times10^{-6}$, which is smaller than $\kappa=\kappa_t+\kappa_r=0.05^2+0.05^2,\ 0.07^2+0.07^2,\ 0.09^2+0.09^2$. 3) As $N$ grows, the ACRs do not continuously increase appreciably. Instead, the ACRs of the IRS-aided communication system are approximately limited by 7.6511 bps/Hz at $\kappa_t=\kappa_r=0.05^2$, 6.6871 bps/Hz at $\kappa_t=\kappa_r=0.07^2$ and 5.9710 bps/Hz at $\kappa_t=\kappa_r=0.09^2$; while those of the DF relay assisted communication system are approximately limited by 4.3237 bps/Hz at $\kappa_t=\kappa_r=0.05^2$, 3.8400 bps/Hz at $\kappa_t=\kappa_r=0.07^2$ and 3.4798 bps/Hz at $\kappa_t=\kappa_r=0.09^2$. The values are consistent with what are computed from $\left.\overline{R_{HWI}}\left(N\right)\right|_{N\rightarrow\infty}=\log_2\left(1+\frac{1}{\kappa_t+\kappa_r}\right)$ and $R_{HWI}^{DF}(N)|_{N\rightarrow\infty}=\frac{1}{2}\log_2\left(1+\frac{2}{\kappa}\right)$. 4) The ACRs of the DF relay assisted communication system increase rapidly when $N<100$, illustrating that the ACR performance of the DF relay can be significantly improved by increasing the quantity of the antennas when $N$ is small.

Then, we compare the utility in (\ref{Utility Expression}) with that in (\ref{eq-DF Utility}) by varying $N$. Figure \ref{Fig-C1} (b) describes the utilities of the IRS and the DF relay in relation to $N$. The observation interval is shrunk ($N\in[1,16]$), for clearly viewing the details on the utility reduction of the DF relay. The results show that: 1) when $\kappa_t$ and $\kappa_r$ grow, the utilities decrease, which confirms that more severe HWI is concomitant with more serious utility degradation. 2) The IRS utility is lower than the DF-relay utility when $N$ is small, and both of them decrease to zero as $N$ grows. This is consistent with what is given in (\ref{N_inf_Utility_Compare}) in \textbf{Lemma 2}.

$\ $

\textit{2) Comparisons by varying $P$}:

\begin{figure*}[!t]
\centering
\subfloat[]{\includegraphics[width=3.2in]{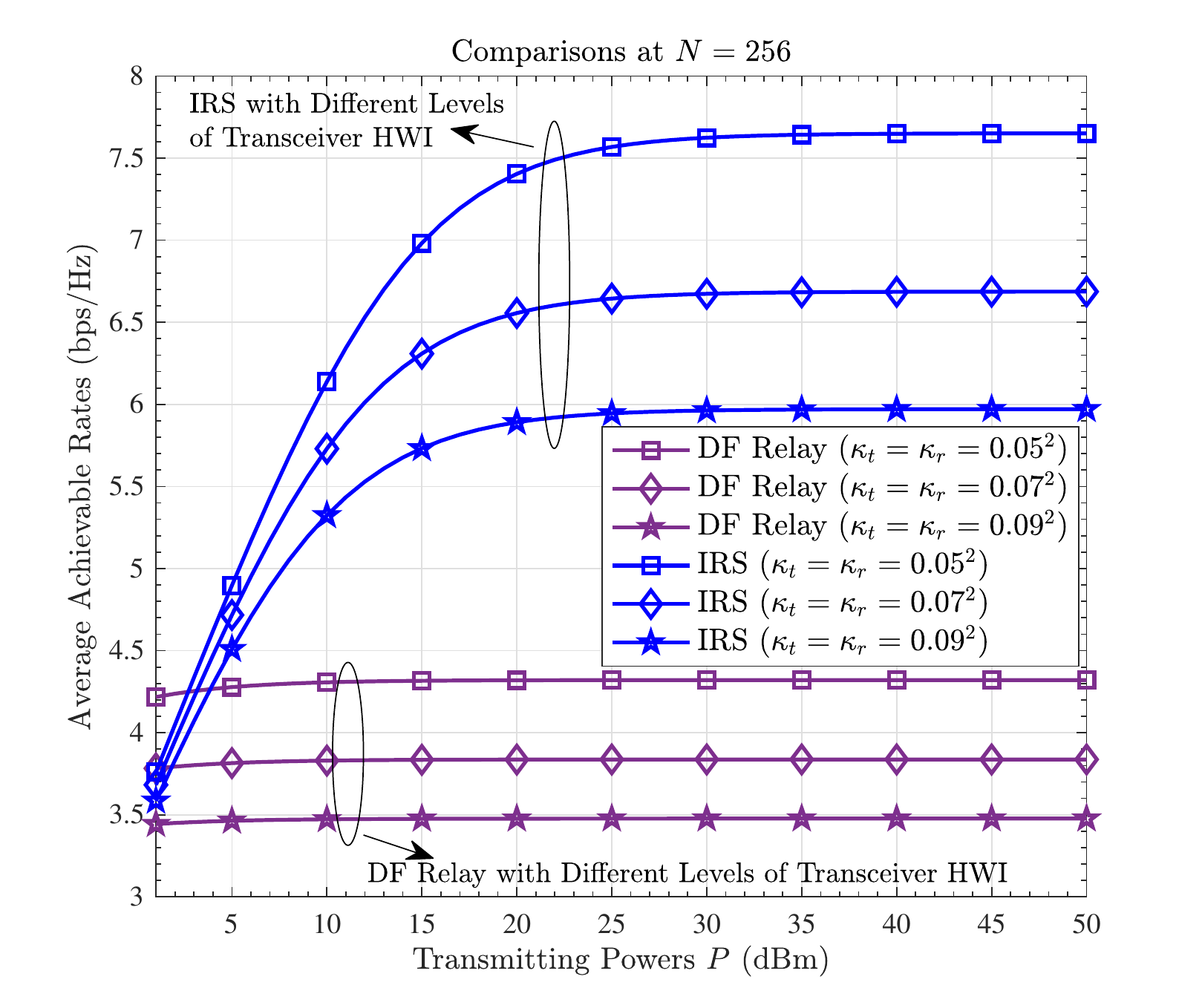}}
\label{Fig-For_P_ACR_Comparison_DF}
\subfloat[]{\includegraphics[width=3.2in]{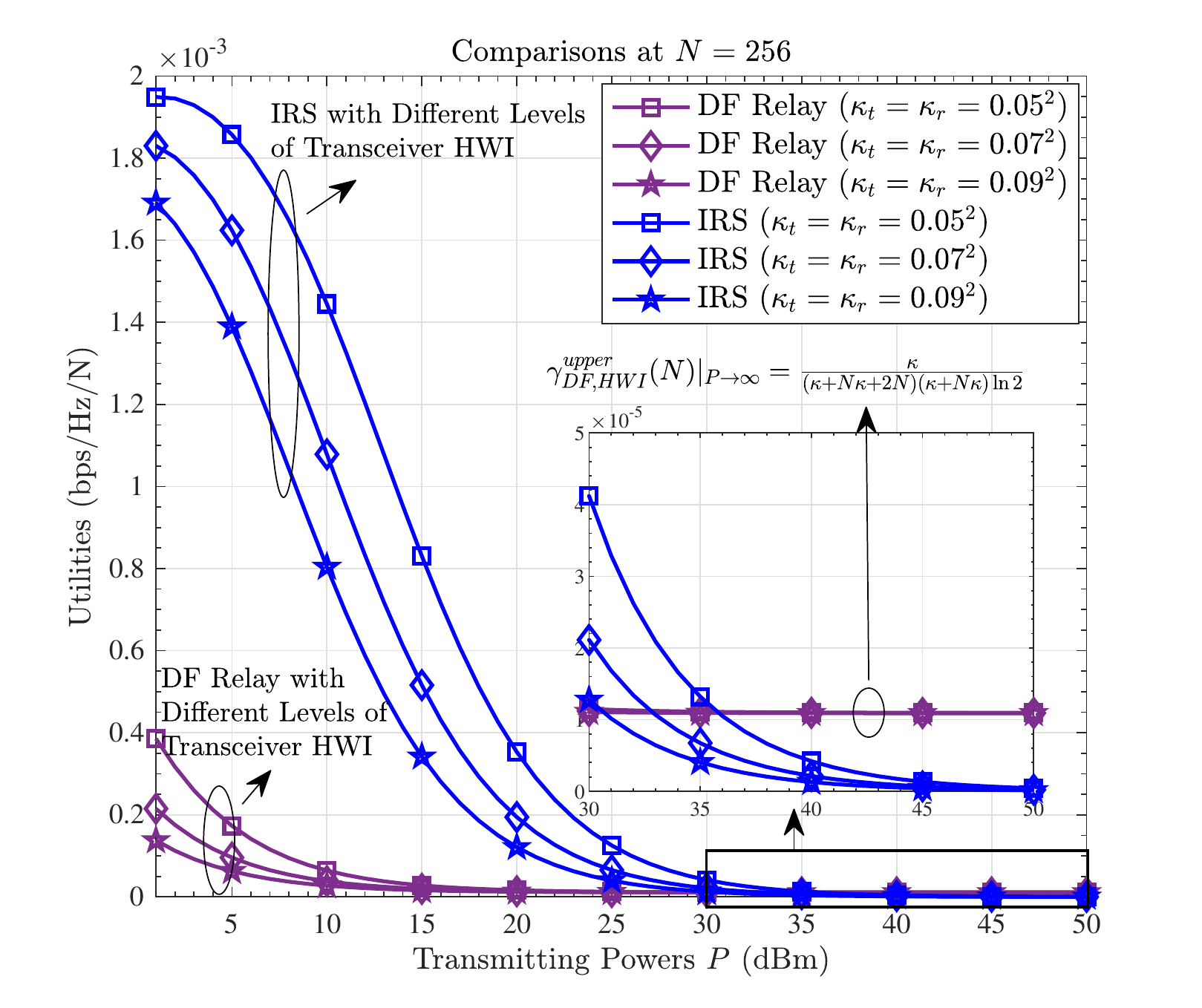}}
\label{Fig-For_P_Utility_Comparison_DF}
\hfil
\caption{{\color{black}Comparisons with DF relay by varying $P$. (a) Average ACRs as functions of $P$, with $N=256$ and $\kappa_t=\kappa_r=0.05^2,0.07^2,0.09^2$. (b) Utilities as functions of $P$, with $N=256$ and $\kappa_t=\kappa_r=0.05^2,0.07^2,0.09^2$.}}
\label{Fig-C2}
\end{figure*}

First, considering the number of the reflecting elements or the DF-relay antennas to be fixed ($N=256$), we compare the average ACR in (\ref{eq2-17}) with that in (\ref{eq2-55}) by varying $P$, and observe the numerical results at $\kappa_t=\kappa_r=0.05^2,0.07^2,0.09^2$. Figure \ref{Fig-C2} (a) plots the average ACRs of the IRS-aided and the DF relay assisted wireless communication systems with respect to $P$, from $P=1$ dBm to $P=50$ dBm. It is indicated that: 1) when $P<2$ dBm and $N=256$, the IRS performs worse than the DF relay if $\kappa_t=\kappa_r=0.05^2,0.07^2$, but better if $\kappa_t=\kappa_r=0.09^2$. This phenomenon reveals that when $P$ is low, the transceiver HWI influences the performance of the DF relay more seriously than the performance of the IRS. 2) The ACRs of the IRS-aided communication system increase faster as $P$ rises, and exceed the ACRs of the DF relay assisted communication system when $P>5$ dBm. 3) Both the ACRs of the IRS-aided and the DF relay assisted communication systems are bounded when $P$ is high. Specifically, the ACRs of the IRS-aided communication system are approximately bounded by 7.6511 bps/Hz at $\kappa_t=\kappa_r=0.05^2$, 6.6871 bps/Hz at $\kappa_t=\kappa_r=0.07^2$ and 5.9710 bps/Hz at $\kappa_t=\kappa_r=0.09^2$; while the ACRs of the DF relay assisted communication system are approximately bounded by 4.3209 bps/Hz at $\kappa_t=\kappa_r=0.05^2$, 3.8372 bps/Hz at $\kappa_t=\kappa_r=0.07^2$ and 3.4770 bps/Hz at $\kappa_t=\kappa_r=0.09^2$. The values coincide with what can be derived from $\overline{R_{HWI}}(N)|_{P\rightarrow\infty}=\log_2\left(1+\frac{1}{\kappa}\right)$ and $R_{HWI}^{DF}(N)|_{P\rightarrow\infty}=\frac{1}{2}\log_2\left(1+\frac{2N}{\kappa+\kappa N} \right)$.

Then, we compare the utility in (\ref{Utility Expression}) with that in (\ref{eq-DF Utility}) by varying $P$. Figure \ref{Fig-C2} (b) depicts the utilities of the IRS and the DF relay with respect to $P$. The results demonstrate that: 1) when $P$ is relatively low, the utilities of the IRS outstrip those of the DF relay, but both of them descend as $P$ grows. 2) The utilities of the IRS decrease to zero, while those of the DF relay converge to certain positive values (around $1.09\times10^{-5}$), which are consistent with what are calculated from $\left.\gamma_{HWI}^{DF}(N)\right|_{P\rightarrow\infty}=\frac{\kappa}{(\kappa+N\kappa+2N)(\kappa+N\kappa)\ln 2}$. These results validate (\ref{P_inf_Utility_Compare_IRS}) and (\ref{P_inf_Utility_Compare_DF}) in \textbf{Lemma 3}, and reveal that in terms of the utility, although the IRS is preferable to the DF relay at a low system power, the DF relay becomes more advantageous if $P$ is significantly high.
}

{\color{black} \section{Conclusion and Future Works}}
In this article, for the purpose of evaluating the performance of the IRS in consideration of {\color{black} the hardware non-ideality at both the IRS and the signal transceivers}, we first analyse the average ACR {\color{black}and the IRS utility} for the IRS-aided SISO communication system, then optimize the IRS phase shifts by converting the original non-convex problem into a SDP problem, {\color{black}subsequently investigate the impact of the channel estimation errors and the residual phase noises on the optimization performance}, and finally compare the IRS with the conventional {\color{black}multiple-antenna} DF relay {\color{black}in terms of the ACR and the utility} in the presence of HWI. The results illustrate that: 1) as the number of the reflecting units grows, the average ACR of the IRS-aided communication system increases, {\color{black}while the utility of the IRS decreases}. 2) The HWI degrades {\color{black}both the ACR and the utility, and it causes more severe ACR reduction when more reflecting elements are equipped.} 3) If the number of the reflecting units is large enough or {\color{black} the transmitting power is sufficiently high}, the IRS can surpass the conventional DF relay in terms of the ACR, {\color{black} although the DF relay is relatively more advantageous in terms of the utility.} Consequently, the IRS is proved to be still an effective facility for data transmission enhancement in the future wireless communication networks with imperfect hardware in the real world.

{\color{black}As in most actual circumstances, the BS is generally equipped with multiple antennas and is responsible for serving multiple users, it is meaningful to dissect the system performance in the IRS-aided multiple-user MISO communication scenario in the presence of HWI. In view of the complex-matrix form of the BS-IRS channel, deriving the closed-form average ACR as a function of the number of the reflecting elements is a challenging task and deserves more effort. In addition, forasmuch as the typical amplify-and-forward (AF) relay is also widely utilized to assist the wireless communication, the insightful theoretical comparison with this conventional approach in the presence of HWI is challenging but worth to be performed in depth as well in the future.}

\appendices
\section{Proof of Theorem 1}
In Appendix A, we will mathematically prove \textbf{Theorem 1} in Section III. 

{\color{black} First, based on (\ref{eq2-16}), the exact average ACR can be derived from
\begin{equation}\label{Exact Ergodic}
\begin{split}
\overline{R_{HWI}}\left(N\right)=&
\mathbb{E}_{\mathbf{\Theta}_E}\left[R_{HWI}\left(N\right)\right]\\=&
\mathbb{E}_{\mathbf{\Theta}_E}\left\{\log_2\left\{1+\frac{P\left|\phi(t)\left(\alpha\mathbf{g}_{IU}^T\mathbf{\Theta}_E\mathbf{g}_{SI}+h_{SU}\right)\right|^2}{P(\kappa_t+\kappa_r)\left|\phi(t)\left(\alpha\mathbf{g}_{IU}^T\mathbf{\Theta}_E\mathbf{g}_{SI}+h_{SU}\right)\right|^2+\sigma_w^2}\right\}\right\}
\end{split}
\end{equation}

However, as illustrated in \cite{D.Li2020(CL)}, it is difficult, if not impossible, to obtain the exact closed-form expression for $\mathbb{E}_{\mathbf{\Theta}_E}\left[R_{HWI}\left(N\right)\right]$.
Therefore, inspired by \cite{D.Li2020(CL), Y.Han2019(TVT), M.Matthaiou2009(TCOM), Q.T.Zhang2005(TWC), S.Sanayei2007(TWC)}, we will also find an approximation to $\mathbb{E}_{\mathbf{\Theta}_E}\left[R_{HWI}\left(N\right)\right]$. 

Fortunately, according to Eq. (35) in \cite{E.Bjornson2013(TCOM)}, which is given by
\begin{equation}\label{eq35_in_[R4.2]}
\mathbb{E}\left\{\log_2\left(1+\frac{x}{y}\right)\right\}\approx
\log_2\left(1+\frac{\mathbb{E}\{x\}}{\mathbb{E}\{y\}}\right)
\end{equation}
a simpler closed-form expression for the average ACR can be achieved by the approximation in (\ref{eq35_in_[R4.2]}). Hence, based on (\ref{eq35_in_[R4.2]}), the $\overline{R_{HWI}}\left(N\right)$ in (\ref{Exact Ergodic}) can be approximated by
\begin{equation}\label{aver_ACR_IT-HWI-approx}
\begin{split}
\overline{R_{HWI}}\left(N\right)\approx&
\log_2\left\{1+\frac{P\mathbb{E}_{\mathbf{\Theta}_E}\left[\left|\phi(t)\left(\alpha\mathbf{g}_{IU}^T\mathbf{\Theta}_E\mathbf{g}_{SI}+h_{SU}\right)\right|^2\right]}{P(\kappa_t+\kappa_r)\mathbb{E}_{\mathbf{\Theta}_E}\left[\left|\phi(t)\left(\alpha\mathbf{g}_{IU}^T\mathbf{\Theta}_E\mathbf{g}_{SI}+h_{SU}\right)\right|^2\right]+\sigma_w^2}\right\} \\
=&\log_2\left(1+\frac{P\mathcal{Q}}{P(\kappa_t+\kappa_r)\mathcal{Q}+\sigma_w^2}\right)
\end{split}
\end{equation}
where $\mathcal{Q}=\mathbb{E}_{\mathbf{\Theta}_E}\left[\left|\phi(t)\left(\alpha\mathbf{g}_{IU}^T\mathbf{\Theta}_E\mathbf{g}_{SI}+h_{SU}\right)\right|^2\right]$.

From (\ref{aver_ACR_IT-HWI-approx}), the problem of deriving the closed-form expression for $\mathbb{E}_{\mathbf{\Theta}_E}\left[R_{HWI}\left(N\right)\right]$ is converted into that for $\mathcal{Q}$, which is expanded into
\begin{equation}\label{Expansion of E}
\begin{split}
\mathcal{Q}=&
\mathbb{E}_{\mathbf{\Theta}_E}\left\{\left[\phi(t)\left(\alpha\mathbf{g}_{IU}^T\mathbf{\Theta}_E\mathbf{g}_{SI}+h_{SU}\right)\right]^*\left[\phi(t)\left(\alpha\mathbf{g}_{IU}^T\mathbf{\Theta}_E\mathbf{g}_{SI}+h_{SU}\right)\right]\right\} \\
=&\mathbb{E}_{\mathbf{\Theta}_E}\left[\left(\alpha\mathbf{g}_{IU}^T\mathbf{\Theta}_E\mathbf{g}_{SI}+h_{SU}\right)^*\left(\alpha\mathbf{g}_{IU}^T\mathbf{\Theta}_E\mathbf{g}_{SI}+h_{SU}\right)\right]
\end{split}
\end{equation}

Subsequently,}
let $\mathbf{G}_{IU}$ and $\mathbf{v}_{E}$ be defined by $\mathbf{G}_{IU}=diag\left(\mathbf{g}_{IU}\right)=\sqrt{\mu_{IU}}\mathbf{I}_N$ and $\mathbf{v}_{E}=\left(e^{j\theta_{E1}},e^{j\theta_{E2}},\ldots,e^{j\theta_{EN}}\right)^T$. Because we have $\mathbf{v}_{E}^T\mathbf{G}_{IU}=\mathbf{g}_{IU}^T\mathbf{\Theta}_E$, $\mathbf{G}_{IU}^*=\mathbf{G}_{IU}$ and $\mathbf{g}_{SI}^*=\mathbf{g}_{SI}$, from (\ref{Expansion of E}) we obtain
\begin{small}
\begin{equation}\label{Convert Q}
\begin{split}
\mathcal{Q}&=\mathbb{E}_{\mathbf{\Theta}_E}\left[{\alpha^2\mathbf{g}}_{SI}^T\mathbf{G}_{IU}^T\mathbf{v}_{E}^*\mathbf{v}_{E}^T\mathbf{G}_{IU}\mathbf{g}_{SI}+\alpha\mathbf{g}_{SI}^T\mathbf{G}_{IU}^T\mathbf{v}_{E}^*h_{SU}+\alpha h_{SU}^*\mathbf{v}_{E}^T\mathbf{G}_{IU}\mathbf{g}_{SI}+||h_{SU}||_2^2\right] \\
&=\alpha^2\mu_{IU}\mu_{SI}\mathbb{E}_{\theta_{Ei}}\left[tr\left(\mathbf{v}_{E}^T\mathbf{\Gamma}_N\mathbf{v}_{E}^*\right)\right]+\alpha\sqrt{\mu_{IU}\mu_{SI}\mu_{SU}}\mathbb{E}_{\theta_{Ei}}\left\{\sum_{i=1}^{N}\left[e^{j{{(\varphi}_{SU}+\theta}_{Ei})}+e^{-j{{(\varphi}_{SU}+\theta}_{Ei})}\right]\right\}+||h_{SU}||_2^2
\end{split}
\end{equation}
\end{small}
where $i=1,2,...,N$. In (\ref{Convert Q}), we can expand $tr\left(\mathbf{v}_{E}^T\mathbf{\Gamma}_N\mathbf{v}_{E}^*\right)$ into 
\begin{small}
\begin{equation}\label{eq-A3}
\begin{split}
tr\left(\mathbf{v}_{E}^T\mathbf{\Gamma}_N\mathbf{v}_{E}^*\right)&=N+\sum_{i\neq1}^{N}e^{j{{(\theta}_{E1}-\theta}_{Ei})}+\sum_{i\neq2}^{N}e^{j{{(\theta}_{E2}-\theta}_{Ei})}+\ldots+\!\!\sum_{i\neq N-1}^{N}e^{j{{(\theta}_{E\left(N-1\right)}-\theta}_{Ei})}+\sum_{i\neq N}^{N}e^{j{{(\theta}_{EN}-\theta}_{Ei})}\\
&=N+2\sum_{i=2}^{N}{cos{{(\theta}_{E1}-\theta}_{Ei})}+2\sum_{i=3}^{N}{cos{{(\theta}_{E2}-\theta}_{Ei})}+\ldots+2\sum_{i=N}^{N}{cos{{(\theta}_{E\left(N-1\right)}-\theta}_{Ei})}\\
&=N+{\mathbf{1M1}}^T
\end{split}
\end{equation} 
\end{small}
where the matrix $\mathbf{M}$ is expressed as
\begin{small}
\begin{equation}\label{eq-A4}
\mathbf{M}=\left(\begin{matrix}2\cos{\left({\theta_{E1}-\theta}_{E2}\right)}&2\cos{\left({\theta_{E2}-\theta}_{E3}\right)}&\begin{matrix}\cdots\ &2\cos{\left({\theta_{E\left(N-1\right)}-\theta}_{EN}\right)}\\\end{matrix}\\2\cos{\left({\theta_{E1}-\theta}_{E3}\right)}&2\cos{\left({\theta_{E2}-\theta}_{E4}\right)}&\begin{matrix}\cdots\ \ \ \ \ \ &\ \ \ \ 0\ \ \ \ \ \ \ \ \ \ \ \ \ \ \ \ \ \\\end{matrix}\\\begin{matrix}\vdots\\\begin{matrix}2\cos{\left({\theta_{E1}-\theta}_{E\left(N-1\right)}\right)}\\2\cos{\left({\theta_{E1}-\theta}_{EN}\right)}\\\end{matrix}\\\end{matrix}&\begin{matrix}\vdots\\\begin{matrix}2\cos{\left({\theta_{E2}-\theta}_{EN}\right)}\\0\\\end{matrix}\\\end{matrix}&\begin{matrix}\begin{matrix}\iddots\ \ \ \ \ \ \ \ \ \ &\ \vdots\ \ \ \ \ \ \ \ \ \ \ \ \ \ \ \ \ \\\end{matrix}\\\begin{matrix}0&\ \ \ \ \ \ \ \ \ \ \ 0\ \ \ \ \ \ \ \ \ \ \ \ \ \ \ \ \\\end{matrix}\\\begin{matrix}0&\ \ \ \ \ \ \ \ \ \ \ 0\ \ \ \ \ \ \ \ \ \ \ \ \ \ \ \ \\\end{matrix}\\\end{matrix}\\\end{matrix}\right)
\end{equation}
\end{small}

We can also utilize Euler formula to expand $\sum_{i=1}^{N}\left[e^{j{{(\varphi}_{SU}+\theta}_{Ei})}+e^{-j{{(\varphi}_{SU}+\theta}_{Ei})}\right]$ and then obtain
$\sum_{i=1}^{N}\left[e^{j{{(\varphi}_{SU}+\theta}_{Ei})}+e^{-j{{(\varphi}_{SU}+\theta}_{Ei})}\right]=2\sum_{i=1}^{N}\cos{\left({\varphi_{SU}+\theta}_{Ei}\right)}$.

As $\theta_{Ei}$, for $i=1,2,...,N$, are random variables which are uniformly distributed on $\left[-\pi/2,\pi/2\right]$, we should calculate the expectations of $2\sum_{i=1}^{N}\cos{\left({\varphi_{SU}+\theta}_{Ei}\right)}$ and $tr\left(\mathbf{v}_{E}^T\mathbf{\Gamma}_N\mathbf{v}_{E}^*\right)$ in order to obtain a statistical average ACR. First, we calculate $\mathbb{E}_{\theta_{Ei}}\left[2\sum_{i=1}^{N}\cos{\left({\varphi_{SU}+\theta}_{Ei}\right)}\right]$ and have
\begin{small}
\begin{equation}\label{eq-A5}
\begin{split}
\mathbb{E}_{\theta_{Ei}}\!\!\left[2\sum_{i=1}^{N}\cos{\left({\varphi_{SU}+\theta}_{Ei}\right)}\right]
&=2\mathbb{E}_{\theta_{Ei}}\!\!\left[\sum_{i=1}^{N}{\cos{\varphi_{SU}}\cos{\theta_{Ei}}}-\sum_{i=1}^{N}{\sin{\varphi_{SU}}\sin{\theta_{Ei}}}\right]\\
&=2N\!\cos{\varphi_{SU}}\!\!\int_{-\frac{\pi}{2}}^{\frac{\pi}{2}}\!\!\!{f\left(\theta_{Ei}\right)\cos{\theta_{Ei}{d\theta}_{Ei}}}-2N\sin{\varphi_{SU}}\int_{-\frac{\pi}{2}}^{\frac{\pi}{2}}{f\left(\theta_{Ei}\right)\sin{\theta_{Ei}{d\theta}_{Ei}}}\\
&=\frac{4}{\pi}N\cos{\varphi_{SU}}
\end{split}
\end{equation}
\end{small}where $f\left(\theta_{Ei}\right)=1/\pi$ is the probability density function of variable $\theta_{Ei}$.

Subsequently, we calculate $\mathbb{E}_{\theta_{Ei}}\left[tr\left(\mathbf{v}_{E}^T\mathbf{\Gamma}_N\mathbf{v}_{E}^*\right)\right]=N+\mathbb{E}_{\theta_{Ei}}\left[{\mathbf{1M1}}^T\right]$. It is notable that the elements in $\mathbf{M}$ are either 0, or $2\cos(\theta_{Ei}-\theta_{Ej})$ for $i<j$. Therefore, let $\delta_{\theta}$ be defined by $\delta_{\theta}=\theta_{Ei}-\theta_{Ej}$. Because $\theta_{Ei}$ obeys uniform distribution on $\left[-\pi/2,\pi/2\right]$, $\delta_{\theta}$ obeys triangular distribution on $[-\pi,\pi]$ whose probability density function is expressed as
\begin{equation}\label{eq-A6}
f\left(\delta_{\theta}\right)=\left\{\begin{matrix}\frac{1}{\pi^2}\delta_{\theta}+\frac{1}{\pi},\ \delta_{\theta}\in\left[-\pi,0\right]\\-\frac{1}{\pi^2}\delta_{\theta}+\frac{1}{\pi},\ \delta_{\theta}\in\left[0,\pi\right]\\\end{matrix}\right.
\end{equation}
Thus, we have
{\color{black} \begin{small}
\begin{equation}\label{eq-A7}
\begin{split}
N+\mathbb{E}_{\theta_{Ei}}\left[{\mathbf{1M1}}^T\right]&=N+\mathbb{E}_{\theta_{Ei}}\left[2\sum_{i<j}^{N}\cos{\left({\theta_{Ei}-\theta}_{Ej}\right)}\right]\\
&=N+N\left(N-1\right)\left[\int_{-\pi}^{0}{\left(\frac{1}{\pi^2}\delta_{\theta}+\frac{1}{\pi}\right)\cos{\left(\delta_{\theta}\right)}d\delta_{\theta}}+\int_{0}^{\pi}{\left(-\frac{1}{\pi^2}\delta_{\theta}+\frac{1}{\pi}\right)\cos{\left(\delta_{\theta}\right)}d\delta_{\theta}}\right]\\
&=N+\frac{1}{\pi^2}N(N-1)\left[\int_{-\pi}^0\delta_{\theta}\cos{(\delta_{\theta})}d\delta_{\theta}-\int_0^{\pi}\delta_{\theta}\cos{(\delta_{\theta})}d\delta_{\theta}\right]\\
&=N+\frac{4}{\pi^2}N(N-1)=\frac{4N^2}{\pi^2}+\left(1-\frac{4}{\pi^2}\right)N
\end{split}
\end{equation}
\end{small}}

By substituting (\ref{eq-A5}) and (\ref{eq-A7}) into (\ref{Convert Q}), and substituting (\ref{Convert Q}) into (\ref{aver_ACR_IT-HWI-approx}), we finally prove (\ref{eq2-17}). {\color{black} Then, by calculating $\gamma_{HWI}(N)=\frac{\partial\overline{R_{HWI}}\left(N\right)}{\partial N}$, we finally prove (\ref{Utility Expression}).}

\section{Proof of Lemma 2}
In Appendix B, we will prove \textbf{Lemma 2} in Section V. 
On the assumption that $\kappa_t=\kappa_r=\frac{1}{2}\kappa$ and $P_1=P_2=P$, {\color{black}when $N\rightarrow\infty$, from (\ref{eq-A}) and (\ref{eq-B}), we have
\begin{equation}
\mathfrak{A}(N)|_{N\rightarrow\infty}=\log_2\left(1+\frac{2}{\kappa}\right)
\end{equation}
\begin{equation}
\mathfrak{B}(N)|_{N\rightarrow\infty}=\log_2\left(1+\frac{1}{\kappa+\frac{\sigma_w^2}{P\mu_{SU}}}+\frac{2}{\kappa}\right)
\end{equation}

Because $1+\frac{2}{\kappa}<1+\frac{1}{\kappa+\frac{\sigma_w^2}{P\mu_{SU}}}+\frac{2}{\kappa}$, according to (\ref{eq2-55}), we have 
\begin{equation}
R_{HWI}^{DF}(N)|_{N\rightarrow\infty}=\frac{1}{2}\mathfrak{A}(N)|_{N\rightarrow\infty}=\frac{1}{2}\log_2\left(1+\frac{2}{\kappa}\right)
\end{equation}

Given $\overline{R_{HWI}}(N)|_{N\rightarrow\infty}$ in (\ref{R_HWI Upper Bound}), we calculate $\overline{R_{HWI}}(N)|_{N\rightarrow\infty}-R_{HWI}^{DF}(N)|_{N\rightarrow\infty}$ and obtain
\begin{equation}
\overline{R_{HWI}}(N)|_{N\rightarrow\infty}-R_{HWI}^{DF}(N)|_{N\rightarrow\infty}
=\frac{1}{2}\log_2\left(1+\frac{1}{\kappa^2+2\kappa}\right)>0
\end{equation}
from which we prove (\ref{N_inf_ACR_Compare}). Then, based on (\ref{Utility Expression}) and (\ref{eq-DF Utility}), we consider $N\rightarrow\infty$ and prove (\ref{N_inf_Utility_Compare}).
}

\section{Proof of Lemma 3}
In Appendix C, we will prove \textbf{Lemma 3} in Section V. 
Similar to the proof of \textbf{Lemma 2} in Appendix B, {\color{black} when $P\rightarrow\infty$, we have
\begin{equation}
\mathfrak{A}(N)|_{P\rightarrow\infty}=\log_2\left(1+\frac{2N}{\kappa+\kappa N} \right)
\end{equation}
\begin{equation}
\mathfrak{B}(N)|_{P\rightarrow\infty}=\log_2\left(1+\frac{1}{\kappa}+\frac{2N}{\kappa+\kappa N} \right)
\end{equation}

Because $1+\frac{2N}{\kappa+\kappa N}<1+\frac{1}{\kappa}+\frac{2N}{\kappa+\kappa N}$, according to (\ref{eq2-55}), we have
\begin{equation}
R_{HWI}^{DF}(N)|_{P\rightarrow\infty}=\frac{1}{2}\mathfrak{A}(N)|_{P\rightarrow\infty}=\frac{1}{2}\log_2\left(1+\frac{2N}{\kappa+\kappa N} \right)
\end{equation}

Based on (\ref{eq2-17}), we have
\begin{equation}
\overline{R_{HWI}}(N)|_{P\rightarrow\infty}=\log_2\left(1+\frac{1}{\kappa}\right)
\end{equation}

Hence, we calculate $\overline{R_{HWI}}(N)|_{P\rightarrow\infty}-R_{HWI}^{DF}(N)|_{P\rightarrow\infty}$ and obtain
\begin{equation}
\overline{R_{HWI}}(N)|_{P\rightarrow\infty}-R_{HWI}^{DF}(N)|_{P\rightarrow\infty}
=\frac{1}{2}\log_2\left\{1+\frac{2\kappa+N+1}{(N+1)\kappa^2+2N\kappa} \right\}>0,\ \ \forall N>0
\end{equation}
from which we prove (\ref{P_inf_ACR_Compare}). Then, on the basis of (\ref{Utility Expression}) and (\ref{eq-DF Utility}), we consider $P\rightarrow\infty$ and prove (\ref{P_inf_Utility_Compare_IRS}) and (\ref{P_inf_Utility_Compare_DF}).
}

\section{Proof of Lemma 4}
In Appendix D, we will prove \textbf{Lemma 4} in Section V. 
{\color{black} It is noted that $\overline{R_{HWI}}(N)$ in (\ref{eq2-17}) is a monotonically increasing function with respect to $N>0$. Thus, its minimum lies on
\begin{equation}
\overline{R_{HWI}}\left(1\right)=\log_2\left\{1+\frac{\beta+\lambda+\mu_{SU}}
{\kappa\left(\beta+\lambda+\mu_{SU}\right)+\frac{\sigma_w^2}{P}}\right\}
\end{equation}
Besides, $R_{HWI}^{DF}(N)$ in (\ref{eq2-55}) is also a monotonically increasing function in relation to $N>0$, which is limited by $R_{HWI}^{DF}(N)|_{N\rightarrow\infty}=\frac{1}{2}\log_2\left(1+\frac{2}{\kappa}\right)$. Therefore, for the IRS to always outperform the DF relay in terms of the ACR for all $N>0$, the following relationship should hold:
\begin{equation}\label{Condition_of_IRS_to_always_outperform}
\overline{R_{HWI}}\left(1\right)>R_{HWI}^{DF}(N)|_{N\rightarrow\infty}
\end{equation}

From (\ref{Condition_of_IRS_to_always_outperform}), after a few manipulations, we obtain
\begin{equation}
\kappa>2\sigma_w^4 \left[P^2(\beta+\lambda+\mu_{SU})^2-2\sigma_w^2P(\beta+\lambda+\mu_{SU})\right]^{-1}
\end{equation}
and consequently prove \textbf{Lemma 4}.
}

\ifCLASSOPTIONcaptionsoff
  \newpage
\fi


%









\end{spacing}


\begin{thebibliography}{1}
  
\bibitem{M.Agiwal2016(CST)}
M. Agiwal, A. Roy and N. Saxena, “Next generation 5G wireless networks: a comprehensive survey,” \emph{IEEE Communications Surveys \& Tutorials}, vol. 18, no. 3, pp. 1617–1655, Third Quarter 2016.

\bibitem{W.Yan2019(WCL)}
W. Yan, X. Yuan and X. Kuai, "Passive beamforming and information transfer via large intelligent surface," \emph{IEEE Wireless Communications Letters}, vol. 9, no. 4, pp. 533–537, Apr. 2020.

\bibitem{X.Lin2015(WCL)}
X. Lin and J. G. Andrews, "Connectivity of millimeter wave networks with multi-hop relaying," \emph{IEEE Wireless Communications Letters}, vol. 4, no. 2, pp. 209–212, Apr. 2015.

\bibitem{W.Sun2018(TWC)}
W. Sun and J. Liu, "2-to-M coordinated multipoint-based uplink transmission in ultra-dense cellular networks,"  \emph{IEEE Transactions on Wireless Communications}, vol. 17, no. 12, pp. 8342–8356, Dec. 2018.

\bibitem{S.Hu2018(TSP)}
S. Hu, F. Rusek and O. Edfors, "Beyond massive MIMO: the potential of data transmission with large intelligent surfaces," \emph{IEEE Transactions on Signal Processing}, vol. 66, no. 10, pp. 2746–2758, May. 2018.

\bibitem{S.Cheng2018(TVT)}
S. Cheng, R. Wang, J. Wu, W. Zhang and Z. Fang, "Performance analysis and beamforming designs of
MIMO AF relaying with hardware impairments," \emph{IEEE Transactions on Vehicular Technology}, vol. 67, no. 7, pp. 6229–6243, Jul. 2018.

\bibitem{P.K.Sharma2016(CL)}
P. K. Sharma and P. K. Upadhyay, "Cognitive relaying with transceiver hardware impairments under interference constraints," \emph{IEEE Communications Letters}, vol. 20, no. 4, pp. 820–823, Apr. 2016.

\bibitem{X.Xia2015(TSP)}
X. Xia \textit{et al}., "Hardware impairments aware transceiver for full-duplex massive MIMO relaying," \emph{IEEE Transactions on Signal Processing}, vol. 63, no. 24, pp. 6565–6580, Dec. 2015.

\bibitem{A.K.Mishra2018(TVT)}
A. K. Mishra and P. Singh, "Performance analysis of opportunistic transmission in downlink cellular DF relay network with channel estimation error and RF impairments," \emph{IEEE Transactions on Vehicular Technology}, vol. 67, no. 9, pp. 9021–9026, Sep. 2018.

\bibitem{O.Ozdogan2019(WCL)}
Ö. Özdogan, E. Björnson and E. G. Larsson, "Intelligent reflecting surfaces: physics, propagation, and pathloss modeling," \emph{IEEE Wireless Communications Letters}, vol. 9, no. 5, pp. 581–585, May. 2020.

\bibitem{Q.Wu2019(ICASSP)}
Q. Wu and R. Zhang, "Beamforming optimization for intelligent reflecting surface with discrete phase shifts," in \emph{Proc. 2019 IEEE International Conference on Acoustics, Speech and Signal Processing (ICASSP)}, Brighton, United Kingdom, May. 2019, pp. 7830–7833.

\bibitem{Z.He2020(WCL)}
Z. He and X. Yuan, "Cascaded channel estimation for large intelligent metasurface assisted massive MIMO," \emph{IEEE Wireless Communications Letters}, vol. 9, no. 2, pp. 210–214, Feb. 2020.

\bibitem{Q.Wu2019(TWC)}
Q. Wu and R. Zhang, "Intelligent reflecting surface enhanced wireless network via joint active and passive beamforming," \emph{IEEE Transactions on Wireless Communications}, vol. 18, no. 11, pp. 5394–5409, Nov. 2019.

\bibitem{Liuyiming(arXiv)}
Y. Liu, E. Liu, R. Wang and Y. Geng, "Beamforming designs and performance evaluations for intelligent reflecting surface enhanced wireless communication system with hardware impairments," May. 2020, [Online]. Available: https://arxiv.org/abs/2006.00664.

\bibitem{C.Huang2019(JSAC)}
C. Huang, R. Mo and C. Yuen, "Reconfigurable intelligent surface assisted multiuser MISO systems exploiting deep reinforcement learning," \emph{IEEE Journal on Selected Areas in Communications}, to be published. DOI: 10.1109/JSAC.2020.3000835.

\bibitem{C.Huang2020(WCM)}
C. Huang, \textit{et al.}, "Holographic MIMO surfaces for 6G wireless networks: opportunities, challenges, and trends," Apr. 2020. [Online]. Available: https://arxiv.org/abs/1911.12296.

\bibitem{C.Huang2018(ICASSP)}
C. Huang, A. Zappone, M. Debbah and C. Yuen, "Achievable rate maximization by passive intelligent mirrors," in \emph{Proc. IEEE International Conference on Acoustics, Speech and Signal Processing (ICASSP)}, Calgary, AB, Canada, Apr. 2018, pp. 3714-3718.

\bibitem{C.Huang2019(TWC)}
C. Huang, A. Zappone, G. C. Alexandropoulos, M. Debbah and C. Yuen, "Reconfigurable intelligent surfaces for energy efficiency in wireless communication," \emph{IEEE Transactions on Wireless Communications}, vol. 18, no. 8, pp. 4157–4170, Aug. 2019.

\bibitem{E.Basar2020(TCOM)}
E. Basar, "Reconfigurable intelligent surface-based index modulation: a new beyond MIMO paradigm for 6G," \emph{IEEE Transactions on Communications}, vol. 68, no. 5, pp. 3187–3196, Feb. 2020.

\bibitem{M.Cui2019(WCL)}
M. Cui, G. Zhang and R. Zhang, "Secure wireless communication via intelligent reflecting surface," \emph{IEEE Wireless Communications Letters}, vol. 8, no. 5, pp. 1410–1414, Oct. 2019.

\bibitem{H.Shen2019(CL)}
H. Shen \textit{et al}., "Secrecy rate maximization for intelligent reflecting surface assisted multi-antenna communications," \emph{IEEE Communications Letters}, vol. 23, no. 9, pp. 1488–1492, Sep. 2019.

\bibitem{Q.Nadeem2020(arXiv)}
Q. Nadeem \textit{et al}., "Intelligent reflecting surface assisted wireless communication: modeling and channel estimation," Dec. 2019. [Online]. Available: https://arxiv.org/abs/1906.02360.

\bibitem{E.Bjornson2020(WCL)}
E. Björnson, Ö. Özdogan and E. G. Larsson, "Intelligent reflecting surface versus decode-and-forward: how large surfaces are needed to beat relaying?" \emph{IEEE Wireless Communications Letters}, vol. 9, no. 2, pp. 244–248, Feb. 2020.

\bibitem{X.Zhang2015(TCOM)}
X. Zhang, M. Matthaiou, M. Coldrey and E. Björnson, "Impact of residual transmit RF impairments
on training-Based MIMO systems," \emph{IEEE Transactions on Communications}, vol. 63, no. 8, pp. 2899–2911, Aug. 2015.

\bibitem{K.Xu2015(IEEE PACRIM)}
K. Xu \textit{et al}., "Achievable rate of full-duplex massive MIMO relaying with hardware impairments," in \emph{Proc. 2015 IEEE Pacific Rim Conference on Communications, Computers and Signal Processing (PACRIM)}, Victoria, BC, Canada, Aug. 2015, pp. 84–89.

\bibitem{T.Schenk2008(book)}
T. Schenk, \emph{RF Imperfections in High-Rate Wireless Systems: Impact and Digital Compensation}. New York, NY, USA: Springer-Verlag, 2008.

\bibitem{E.Bjornson2014(TIT)}
E. Björnson, J. Hoydis, M. Kountouris and M. Debbah, "Massive MIMO systems with non-ideal hardware: Energy efficiency, estimation, and capacity limits," \emph{IEEE Transactions on Information Theory}, vol. 60, no. 11, pp. 7112–7139, Nov. 2014.

\bibitem{E.Bjornson2013(CL)}
E. Björnson, P. Zetterberg, M. Bengtsson and B. Ottersten, "Capacity limits and multiplexing gains of MIMO channels with transceiver impairments," \emph{IEEE Communications Letters}, vol. 17, no. 1, pp. 91–94, Jan. 2013.

\bibitem{Q.Zhang2018(TWC)}
Q. Zhang, T. Q. S. Quek and S. Jin, "Scaling analysis for massive MIMO systems with hardware impairments in Rician fading," \emph{IEEE Transactions on Wireless Communications}, vol. 17, no. 7, pp. 4536–4549, Jul. 2018.

{\color{black}
\bibitem{E.Bjornson2015(TWC)}
E. Björnson, M. Matthaiou and M. Debbah, "Massive MIMO with non-Ideal arbitrary arrays: Hardware scaling laws and circuit-aware design," \emph{IEEE Transactions on Wireless Communications}, vol. 14, no. 8, pp. 4353–4368, Aug. 2015.
}

\bibitem{S.Hu2018(GlobalCom)}
S. Hu, F. Rusek and O. Edfors, "Capacity degradation with modeling hardware impairment in large intelligent surface," in \emph{Proc. 2018 IEEE Global Communications Conference (GLOBECOM)}, Abu Dhabi, United Arab Emirates, United Arab Emirates, Dec. 2018, pp. 1–6.

\bibitem{M.A.Badiu2020(WCL)}
M. A. Badiu and J. P. Coon, "Communication through a large reflecting surface with phase errors," \emph{IEEE Wireless Communications Letters}, vol. 9, no. 2, pp. 184–188, Feb. 2020.

{\color{black} \bibitem{D.Li2020(CL)}
D. Li, “Ergodic capacity of intelligent reflecting surface-assisted communication systems with phase errors,” \emph{IEEE Communications Letters}, vol. 24, no. 8, pp. 1646-1650, Aug. 2020.}

\bibitem{B.Zheng2020(WCL)}
B. Zheng and R. Zhang, "Intelligent reflecting surface-enhanced OFDM: Channel estimation and reflection optimization," \emph{IEEE Wireless Communications Letters}, vol. 9, no. 4, pp. 518–522, Apr. 2020.

{\color{black} \bibitem{L.Wei2020(Arxiv)}
L. Wei, C. Huang, G. C. Alexandropoulos, C. Yuen, Z. Zhang and M. Debbah, "Channel estimation for RIS-empowered multi-user MISO wireless communications," Aug. 2020. [Online]. Available: https://arxiv.org/abs/2008.01459.}

{\color{black} \bibitem{L.Wei2020(IEEE SAM)}
L. Wei, C. Huang, G. C. Alexandropoulos and C. Yuen, “Parallel factor decomposition channel estimation in RIS-assisted
multi-user MISO communication,” in \emph{Proc. IEEE 11th Sensor Array and Multichannel Signal Processing Workshop (SAM)}, Hangzhou, China, Jun. 2020, pp. 1–6.}

{\color{black} \bibitem{L.Liu2014(TSP)}
L. Liu, R. Zhang and K.-C. Chua, “Secrecy wireless information and power transfer with MISO beamforming,” \emph{IEEE Transactions on Signal Processing}, vol. 62, no. 7, pp. 1850-1863, Apr. 2014.}

\bibitem{A.Nemirovski2008(ActaNumerica)}
I. Polik and T. Terlaky, “Interior point methods for nonlinear optimization,” in \emph{Nonlinear Optimization}, 1st ed., G. Di Pillo and F. Schoen, Eds. Berlin, Germany: Springer, 2010, ch. 4.

\bibitem{J.N.Laneman2004(TIT)}
J. N. Laneman, D. N. C. Tse, and G. W. Wornell, “Cooperative diversity in wireless networks: efficient protocols and outage behavior,” \emph{IEEE Transactions on Information Theory}, vol. 50, no. 12, pp. 3062–3080, Dec. 2004.

{\color{black} \bibitem{J.Zhang2020(CL)}
J. Zhang, Y. Zhang, C. Zhong and Z. Zhang, "Robust design for intelligent reflecting surfaces assisted MISO systems," \emph{IEEE Communications Letters}, vol. 24, no. 10, pp. 2353-2357, Oct. 2020.}

{\color{black} \bibitem{Y.Han2019(TVT)}
Y. Han, W. Tang, S. Jin, C.-K. Wen and X. Ma, "Large intelligent surface-assisted wireless communication exploiting statistical CSI," \emph{IEEE Transactions on Vehicular Technology}, vol. 68, no. 8, pp. 8238-8242, Aug. 2019.}

{\color{black} \bibitem{M.Matthaiou2009(TCOM)}
M. Matthaiou, Y. Kopsinis, D. I. Laurenson and A. M. Sayeed, "Ergodic capacity upper bound for dual MIMO Ricean systems: Simplified derivation and asymptotic tightness," \emph{IEEE Transactions on Communications}, vol. 57, no. 12, pp. 3589-3596, Dec. 2009.}

{\color{black} \bibitem{Q.T.Zhang2005(TWC)}
Q. T. Zhang, X. W. Cui and X. M. Li, "Very tight capacity bounds for MIMO-correlated Rayleigh-fading channels," \emph{IEEE Transactions on Wireless Communications}, vol. 4, no. 2, pp. 681-688, Mar. 2005.}

{\color{black} \bibitem{S.Sanayei2007(TWC)}
S. Sanayei and A. Nosratinia, "Opportunistic beamforming with limited feedback," \emph{IEEE Transactions on Wireless Communications}, vol. 6, no. 8, pp. 2765-2771, Aug. 2007.}

{\color{black} \bibitem{E.Bjornson2013(TCOM)}
E. Björnson, M. Matthaiou and M. Debbah, “A new look at dual-hop relaying: Performance limits with hardware impairments,” \emph{IEEE Transactions on Communications}, vol. 61, no. 11, pp. 4512-4525, Nov. 2013.}

{\color{black}\bibitem{Charnes1962}
A. Charnes and W. W. Cooper, “Programming with linear fractional functions,” \emph{Naval Research Logistics}, vol. 9, pp. 181–186, Dec. 1962.}

\end{thebibliography}
\end{document}